\documentclass{article}

\usepackage[colorlinks,allcolors=blue]{hyperref}
\usepackage{amssymb,amsmath,amsthm,mathrsfs}
\usepackage{venturis2}
\usepackage[T1]{fontenc}
\usepackage{enumerate}
\usepackage{graphicx}
\usepackage{cite}

\def\B{\mathscr B}
\def\C{\mathbb C}
\def\d{\mathrm{d}}
\def\DD{\mathscr D}
\def\dom{\mathcal D}
\def\G{\mathcal G}
\def\H{\mathcal H}
\def\K{\mathscr K}

\def\ltwoloc{\mathop{\mathrm{L}^2_{\rm loc}}\nolimits}

\def\linf{\mathop{\mathrm{L}^\infty}\nolimits}
\def\N{\mathbb N}
\def\NN{\mathfrak N}
\def\R{\mathbb R}
\def\S{\mathbb S}

\def\Tau{\mathcal T}
\def\U{\mathrm U}
\def\UU{\mathscr U}
\def\V{\mathtt V}
\def\X{\mathtt X}
\def\Z{\mathbb Z}

\def\e{\mathop{\mathrm{e}}\nolimits}
\def\Ran{\mathop{\mathrm{Ran}}\nolimits}
\def\Span{\mathop{\mathrm{Span}}\nolimits}
\def\supp{\mathop{\mathrm{supp}}\nolimits}
\def\tsum{\mathop{\textstyle\sum}\nolimits}
\DeclareMathOperator*{\slim}{s\hspace{0.1pt}-\hspace{0.1pt}lim}

\newtheorem{Theorem}{Theorem}[section]
\newtheorem{Remark}[Theorem]{Remark}
\newtheorem{Lemma}[Theorem]{Lemma}
\newtheorem{Corollary}[Theorem]{Corollary}
\newtheorem{Proposition}[Theorem]{Proposition}
\newtheorem{Assumption}[Theorem]{Assumption}
\newtheorem{Definition}[Theorem]{Definition}

\begin{document}


\title{Spectral and scattering properties of quantum walks on homogenous trees of odd
degree}

\author{R. Tiedra de Aldecoa\footnote{Partially supported by the Chilean Fondecyt
Grant 1210003.}}

\date{\small}
\maketitle
\vspace{-1cm}

\begin{quote}
\begin{itemize}
\item[] Facultad de Matem\'aticas, Pontificia Universidad Cat\'olica de Chile,\\
Av. Vicu\~na Mackenna 4860, Santiago, Chile\\
E-mail: rtiedra@mat.uc.cl
\end{itemize}
\end{quote}


\begin{abstract}
For unitary operators $U_0,U$ in Hilbert spaces $\H_0,\H$ and identification operator
$J:\H_0\to\H$, we present results on the derivation of a Mourre estimate for $U$
starting from a Mourre estimate for $U_0$ and on the existence and completeness of the
wave operators for the triple $(U,U_0,J)$. As an application, we determine spectral
and scattering properties of a class of anisotropic quantum walks on homogenous trees
of odd degree with evolution operator $U$. In particular, we establish a
Mourre estimate for $U$, obtain a class of locally $U$-smooth operators, and prove
that the spectrum of $U$ covers the whole unit circle and is purely absolutely
continuous, outside possibly a finite set where $U$ may have eigenvalues of finite
multiplicity. We also show that (at least) three different choices of free evolution
operators $U_0$ are possible for the proof of the existence and completeness of the
wave operators.
\end{abstract}


\smallskip


\tableofcontents

\section{Introduction and main results}\label{section_intro}
\setcounter{equation}{0}

Recent years have seen a surge of research activity on discrete-time systems described
by unitary evolution operators. CMV matrices, quantum walks, Koopman operators of
dynamical systems and Floquet operators are classes of such systems having received a
lot of attention. This surge of activity has in turn motivated various researchers to
develop, or adapt from the self-adjoint setup, mathematical tools suited for the
spectral and scattering analysis of unitary operators. Among these tools is the Mourre
theory for unitary operators, which has been first introduced in \cite{ABCF06} and
then further developed in several papers such as
\cite{ABC15_1,ABC15_2,Bou_2013,FRT13,RST_2018} (see also the precursor works
\cite{Kat68,Put67}).

When the evolution operator $U$ of the system under study is nontrivial, it is often
better to start by determining properties of a simpler evolution operator $U_0$
describing the free dynamics, and then infer from it properties of $U$. This is the
core idea of perturbation theory for linear operators \cite{Kat95}. However, the
operator $U_0$ may be defined in a different Hilbert space than the operator $U$, as
for instance when the system has a multichannel structure. In such a case, the
perturbation theory has to be adapted accordingly. In the first part of the paper, we
collect results in this direction for Mourre theory and scattering theory for unitary
operators. Our results are either new or extensions of abstract results of
\cite{RST_2018,RST_2019}, and they can be considered as a unitary analogue of the
results of \cite{RT13_2} in the self-adjoint case.

In the second part of the paper, we apply our abstract results to quantum walks on
homogenous trees of odd degree $d\ge3$ as defined in \cite{HJ_2014,JM_2014} (but see
also \cite{CHKS_2009,DRMBNK_2011} for other definitions of quantum walks on trees).
Motivated by recent works on quantum walks on $\Z$
\cite{ABJ_2015,RST_2018,RST_2019,Suz_2016}, we consider quantum walks with a
position-dependent coin admiting a limit at infinity on each main branch of the tree.
Since $\Z$ is nothing else but a homogenous tree of degree $2$, these quantum walks on
homogenous trees of degree $d$ are to some extent a generalisation of the anisotropic
quantum walks on $\Z$ considered in \cite{RST_2018,RST_2019}.

\begin{figure}[h]
\centering
\includegraphics[width=240pt]{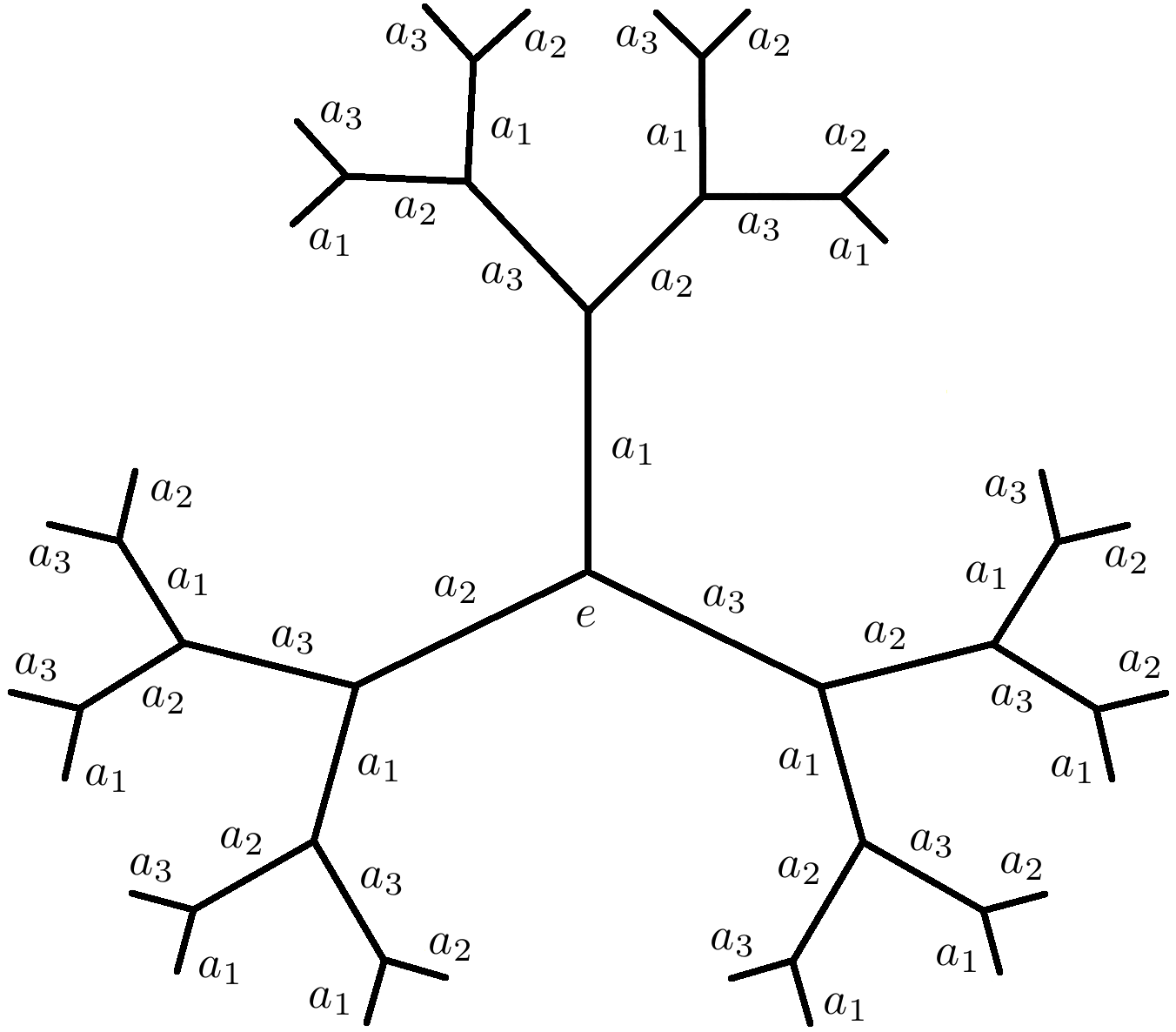}
\caption{Homogenous tree of degree $3$ with generators $a_1,a_2,a_3$}
\label{fig_tree}
\end{figure}

Here is a more detailed description of the paper. In Section \ref{sec_one}, we gather
the needed results on Mourre theory for unitary operators in one Hilbert space. Given
a unitary operator $U$ and a self-adjoint operator $A$ in a Hilbert space $\H$, we
recall the definitions of commutator classes relative to $A$, locally $U$-smooth
operators, conjugate operator for $U$, and Mourre estimate for $U$. We also recall the
two main consequences (under some regularity assumptions) of a Mourre estimate for
$U$: the existence of locally $U$-smooth operators and spectral properties of $U$
(Theorems \ref{thm_locally_smooth} and \ref{thm_spec_prop}).

In Section \ref{sec_two}, we consider a second unitary operator $U_0$ in a second
Hilbert space $\H_0$, and assume that an identification operator $J:\H_0\to\H$ is
given. Then we present an improved version of the results of \cite{RST_2018} on the
construction of a conjugate operator $A$ for $U$ starting from a conjugate operator
$A_0$ for $U_0$. In particular, we determine conditions guaranteeing that the operator
$A=JA_0J^*$ is a conjugate operator for $U$ if the operator $A_0$ is a conjugate
operator for $U_0$ (Theorem \ref{thm_rho_bis}).

In Section \ref{sec_scatt}, we collect results on the existence and completeness
under smooth perturbations of the wave operators for the triple $(U,U_0,J)$. We recall
the intertwinning property of the wave operators (Lemma \ref{lemma_intertwinning}),
present criteria for the $J$-completeness of the wave operators (Theorem
\ref{thm_wave} and Corollary \ref{corol_wave}), prove a chain rule for the wave
operators (Lemma \ref{lemma_chain}), and finally present a criterion for the
completeness, in the usual sense, of the wave operators (Theorem \ref{thm_wave_bis}).

In Section \ref{section_tree}, we apply the results that precede to quantum walks on
homogeneous trees of odd degree $d\ge3$. We consider quantum walks with evolution operator
$U=SC$, where $S$ is a shift and $C$ a position-dependent coin that converge at
infinity, on each main branch of the tree, to a constant diagonal matrix $C_i$,
$i=1,\ldots,d$ (Assumption \ref{ass_short}). This provides evolution operators $U_i:=SC_i$
describing the asymptotic behaviour of $U$ on each main branch of the tree. This also
motivates to choose $U_0:=\bigoplus_{k=1}^dU_k$ as a free evolution operator, and
suggests the definition of the identification operator $J$. In Section
\ref{section_free}, we construct a conjugate operator for $U_0$ and determine the
spectral properties of $U_0$. Using an approach motivated by classical mechanics, we
contruct a conjugate operator $A_0$ satisfying the homogeneous Mourre estimate
$$
U_0^{-1}[A_0,U_0]=2.
$$
This relation, together with Mackey's imprimitivity theorem, implies that $U_0$ is
unitarily equivalent to a multiplication operator with purely absolutely continuous
spectrum covering the whole unit circle (Proposition \ref{prop_spec_U_0}). In Section
\ref{section_full}, we show that the operator $A=JA_0J^*$ is a conjugate operator for
$U$ (Lemma \ref{lemma_A}), establish a Mourre estimate for $U$ (Proposition
\ref{prop_Mourre}), obtain a class of locally $U$-smooth operators (Theorem
\ref{thm_smooth_walk}), and show that $U$ has at most finitely many eigenvalues, each
one of finite multiplicity, and no singular continuous spectrum (Theorem
\ref{thm_spec_U}). To our knowledge this is the first spectral result of this type for
a class of quantum walks on homogeneous trees of degree $d\ge3$ with position-dependent
coin (see \cite[Thm.~4.5]{JM_2014} for a result on a class of quantum walks on rooted
trees of degree $3$ with constant coin). Finally, we establish in Section
\ref{sec_wave} the existence and completeness of the wave operators for the triple
$(U,U_0,J)$ (Theorem \ref{thm_comp_1}). As a by-product, we infer that the absolutely
continuous spectrum of $U$ covers the whole unit circle (Corollary \ref{cor_spec_U}).
Interestingly enough, we also show that two operators different from $U_0$ can be used
as a free evolution operator for $U$, and establish the existence and completeness of
the wave operators in these cases too (Corollary \ref{cor_comp_2} and Theorem
\ref{thm_comp_3}). Of course, each choice of free evolution operator comes with its
pros and cons, see Remark \ref{rem_3_free}.

We point out that our proof of absence of singular continuous spectrum is a key
ingredient to establish in the future a weak limit theorem for $U$. Weak limit
theorems are results relating scattering properties of quantum walks to probabilistic
properties of the corresponding classical random walks. Typically, these theorems show
that if $\X_n$ is the random variable for the position of a quantum walker at time
$n\in\Z$, then $\X_n/n$ converges in law to a random variable $\V$ as $n\to\infty$
with probability distribution $\mu_\V$ given in terms of scattering quantities of the
evolution operator of the quantum walk. See \cite{Kon02,Kon05,RST_2019,Suz_2016} for
examples of weak limit theorems for quantum walks on $\Z$.

To conclude, we emphasize that this work leaves open for future investigations
various other interesting problems about quantum walks on homogeneous trees of odd
degree. The existence of asymptotic velocity operators, the description of the initial
subspaces of the wave operators, the generalisation to coin operators converging to
constant coin along more refined partitions of the tree, the extension to rooted
trees, and the generalisation to coin operators admiting non-diagonal matrix limits at
infinity are examples of problems that could be investigated. See the comments
(i)-(iv) at the end of Section \ref{sec_wave} for more details.

\bigskip
\noindent
{\bf Acknowledgements:} The author thanks the anonymous referees for their suggestions
which helped improve the redaction of various definitions and motivated the extension
of the results to homogenous trees of odd degree $>3$.

\bigskip
\noindent
{\bf Notations:} $\N:=\{0,1,2,\ldots\}$ is the set of natural numbers, $\S^1$ the
complex unit circle, and $\U(n)$ the group of $n\times n$ unitary matrices. Given a
Hilbert space $\H$, we write $\|\cdot\|_\H$ for its norm and
$\langle\cdot,\cdot\rangle_\H$ for its scalar product (linear in the first argument).
Given two Hilbert spaces $\H_1,\H_2$, we write $\B(\H_1,\H_2)$ (resp. $\K(\H_1,\H_2)$,
$S_2(\H_1,\H_2)$) for the set of bounded (resp. compact, Hilbert-Schmidt) operators
from $\H_1$ to $\H_2$. We also write $\|\cdot\|_{\B(\H_1,\H_2)}$ for the norm of
$\B(\H_1,\H_2)$, and use the shorthand notations $\B(\H_1):=\B(\H_1,\H_1)$,
$\K(\H_1):=\K(\H_1,\H_1)$, and $S_2(\H_1):=S_2(\H_1,\H_1)$.

\section{Mourre theory in one Hilbert space}\label{sec_one}
\setcounter{equation}{0}

A unitary operator $U$ in a Hilbert space $\H$ is a surjective isometry, that is, an
operator $U\in\B(\H)$ satisfying $U^*U=UU^*=1$. Since $U^*U=UU^*=1$, the spectral
theorem for normal operators implies that $U$ admits exactly one complex spectral
family $E_U$, with support $\supp(E_U)\subset\S^1$, such that
$U=\int_\C\d E_U(z)\;\!z$. The support $\supp(E_U)$ is the set of points of
non-constancy of $E_U$, and it coincides with the spectrum $\sigma(U)$ of $U$
\cite[Thm.~7.34(a)]{Wei80}. One can associate in a canonical way a real spectral
family to the complex spectral family $E_U$. Indeed, if we let $E^U$ be the spectral
measure corresponding to the spectral family $E_U$, then the family $\widetilde E_U$
defined by
$$
\widetilde E_U(t):=
\begin{cases}
0 & \hbox{if $t<0$}\\
E^U\big(\{\e^{is}\mid s\in[0,t]\}\big) & \hbox{if $t\in[0,2\pi)$}\\
1 & \hbox{if $t\ge2\pi$,}
\end{cases}
$$
is a real spectral family with support $\supp(\widetilde E_U)\subset[0,2\pi]$ which
satisfies $U=\int_\R\d\widetilde E_U(t)\e^{it}$ \cite[Thm.~7.36]{Wei80}. Since
$\widetilde E_U$ is a real spectral family, the corresponding real spectral measure
$\widetilde E^U$ admits a decomposition
$$
\widetilde E^U=\widetilde E^U_{\rm p}+\widetilde E^U_{\rm sc}+\widetilde E^U_{\rm ac},
$$
with $\widetilde E^U_{\rm p}$, $\widetilde E^U_{\rm sc}$, $\widetilde E^U_{\rm ac}$
the pure point, the singular continuous, and the absolutely continuous components of
$\widetilde E_U$, respectively. The corresponding subspaces
$\H_{\rm p}(U):=P_{\rm p}(U)\H$, $\H_{\rm sc}(U):=P_{\rm sc}(U)\H$,
$\H_{\rm ac}(U):=P_{\rm ac}(U)\H$ with $P_\star(U):=\widetilde E^U_\star(\R)$
($\star=$ p,\,sc,\,ac) provide an orthogonal decomposition which reduces the operator
$U:$
$$
\H=\H_{\rm p}(U)\oplus\H_{\rm sc}(U)\oplus\H_{\rm ac}(U).
$$
The sets $\sigma_{\rm p}(U):=\sigma(U|_{\H_{\rm p}(U)})$,
$\sigma_{\rm sc}(U):=\sigma(U|_{\H_{\rm sc}(U)})$, and
$\sigma_{\rm ac}(U):=\sigma(U|_{\H_{\rm ac}(U)})$ are called pure point spectrum,
singular continuous spectrum, and absolutely continuous spectrum of $U$, respectively,
and the set $\sigma_{\rm c}(U):=\sigma_{\rm sc}(U)\cup\sigma_{\rm ac}(U)$ is called
the continuous spectrum of $U$. We also use the notation $\sigma_{\rm ess}(U)$ for the
essential spectrum of $U$.

If $\G$ is an auxiliary Hilbert space, then an operator $T\in\B(\H,\G)$ is called
locally $U$-smooth on a Borel set $\Theta\subset\S^1$ if there exists $c_\Theta\ge0$
such that
\begin{equation}\label{def_U_smooth}
\sum_{n\in\Z}\big\|TU^nE^U(\Theta)\varphi\big\|_\G^2
\le c_\Theta\|\varphi\|_{\H_0}^2\quad\hbox{for all $\varphi\in\H$,}
\end{equation}
and $T$ is called $U$-smooth if \eqref{def_U_smooth} is satisfied with $\Theta=\S^1$.
The condition \eqref{def_U_smooth} is invariant under rotation by $\theta\in\S^1$ in
the sense that if $T$ is locally $U$-smooth on $\Theta$, then $T$ is locally
$(\theta U)$-smooth on $\theta\Theta$ since
$$
\big\|T(\theta U)^nE^{\theta U}(\theta\Theta)\varphi\big\|_\G
=\big\|TU^nE^U(\Theta)\varphi\big\|_\G
\quad\hbox{for all $\varphi\in\H$.}
$$
As in the self-adjoint case, the existence of a locally $U$-smooth operator can be
formulated in various equivalent ways (see \cite[Thm.~7.1]{ABC15_1}). An important
consequence of the existence of a locally $U$-smooth operator $T$ on $\Theta$ is the
inclusion $\overline{E^U(\Theta)T^*\G^*}\subset\H_{\rm ac}(U)$, where $\G^*$ is the
adjoint space of $\G$ (see \cite[Thm.~2.1 \& Def.~2.2]{ABCF06}).

We now recall some results on Mourre theory for unitary operators in one Hilbert
space, starting with definitions and results borrowed from
\cite{ABG96,FRT13,RST_2018}. Let $S\in\B(\H)$ and let $A$ be a self-adjoint operator
in $\H$ with domain $\dom(A)$. For any $k\in\N$, we say that $S$ belongs to $C^k(A)$,
with notation $S\in C^k(A)$, if the map $\R\ni t\mapsto\e^{-itA}S\e^{itA}\in\B(\H)$ is
strongly of class $C^k$. In the case $k=1$, one has $S\in C^1(A)$ if and only if the
quadratic form
$$
\dom(A)\ni\varphi\mapsto\langle A\;\!\varphi,S\varphi\rangle_\H
-\langle\varphi,SA\;\!\varphi\rangle_\H\in\C
$$
is continuous for the topology induced by $\H$ on $\dom(A)$. The operator associated
to the continuous extension of the form is denoted by $[A,S]\in\B(\H)$, and it
verifies
$$
[A,S]=\slim_{t\to0}\;\![A(t),S]
\quad\hbox{with}\quad A(t):=\tfrac1{it}(\e^{it A}-1)\in\B(\H),
\quad t\in\R\setminus\{0\}.
$$

Three regularity conditions slightly stronger than $S\in C^1(A)$ are defined as
follows. $S$ belongs to $C^{1,1}(A)$, with notation $S\in C^{1,1}(A)$, if
$$
\int_0^1\tfrac{\d t}{t^2}\;\!
\big\|\e^{-itA}S\e^{itA}+\e^{itA}S\e^{-itA}-2S\big\|_{\B(\H)}<\infty.
$$
$S$ belongs to $C^{1+0}(A)$, with notation $S\in C^{1+0}(A)$, if $S\in C^1(A)$ and
$$
\int_0^1\tfrac{\d t}t\;\!\big\|\e^{-itA}[A,S]\e^{itA}-[A,S]\big\|_{\B(\H)}<\infty.
$$
$S$ belongs to $C^{1+\varepsilon}(A)$ for some $\varepsilon\in(0,1)$, with notation
$S\in C^{1+\varepsilon}(A)$, if $S\in C^1(A)$ and
$$
\big\|\e^{-itA}[A,S]\e^{itA}-[A,S]\big\|_{\B(\H)}
\le{\rm Const.}\;\!t^\varepsilon\quad\hbox{for all $t\in(0,1)$.}
$$
As banachisable topological vector spaces, these sets satisfy the continuous
inclusions \cite[Sec.~5.2]{ABG96}
$$
C^2(A)\subset C^{1+\varepsilon}(A)\subset C^{1+0}(A)\subset C^{1,1}(A)\subset C^1(A)
\subset C^0(A)\equiv\B(\H).
$$

Now, let $U$ be a unitary operator with $U\in C^1(A)$. For $S,T\in\B(\H)$, we write
$T\gtrsim S$ if there exists an operator $K\in\K(\H)$ such that $T+K\ge S$, and for
$\theta\in\S^1$ and $\varepsilon>0$ we set
$$
\Theta(\theta;\varepsilon)
:=\big\{\theta'\in\S^1\mid|\arg(\theta-\theta')|<\varepsilon\big\}
\quad\hbox{and}\quad
E^U(\theta;\varepsilon):=E^U\big(\Theta(\theta;\varepsilon)\big).
$$
With these notations at hand, we can define functions
$\varrho^A_U:\S^1\to(-\infty,\infty]$ and
$\widetilde\varrho^A_U:\S^1\to(-\infty,\infty]$ by
\begin{align*}
\varrho^A_U(\theta)
&:=\sup\big\{a\in\R\mid\exists\;\!\varepsilon>0~\hbox{such that}~
E^U(\theta;\varepsilon)U^{-1}[A,U]E^U(\theta;\varepsilon)
\ge a\;\!E^U(\theta;\varepsilon)\big\},\\
\widetilde\varrho^A_U(\theta)
&:=\sup\big\{a\in\R \mid\exists\;\!\varepsilon>0~\hbox{such that}~
E^U(\theta;\varepsilon)U^{-1}[A,U]E^U(\theta;\varepsilon)
\gtrsim a\;\!E^U(\theta;\varepsilon)\big\}.
\end{align*}
In applications, $\widetilde\varrho^A_U$ is more convenient than $\varrho^A_U $ since
it is defined in terms of a weaker positivity condition (positivity up to compact
terms). A simple argument shows that $\widetilde\varrho^A_U(\theta)$ can be defined in
an equivalent way by
$$
\widetilde\varrho^A_U(\theta)
=\sup\big\{a\in\R\mid\exists\;\!\eta\in C^\infty(\S^1,\R)~\hbox{such that}
~\eta(\theta)\ne0~\hbox{and}~\eta(U)U^{-1}[A,U]\eta(U)\gtrsim a\;\!\eta(U)^2\big\}.
$$

Further properties of the functions $\widetilde\varrho^A_U$ and $\varrho^A_U$ are
recalled in the following lemma.

\begin{Lemma}[Lemma 3.3 of \cite{RST_2018}]\label{lemma_properties}
Let $U$ be a unitary operator in $\H$ and let $A$ be a self-adjoint operator in $\H$
with $U\in C^1(A)$.
\begin{enumerate}[(a)]
\item The function $\varrho^A_U:\S^1\to(-\infty,\infty]$ is lower semicontinuous,
and $\varrho^A_U(\theta)<\infty$ if and only if $\theta\in\sigma(U)$.
\item The function $\widetilde\varrho^A_U:\S^1\to(-\infty,\infty]$ is lower
semicontinuous, and $\widetilde\varrho^A_U(\theta)<\infty$ if and only if
$\theta\in\sigma_{\rm ess}(U)$.
\item $\widetilde\varrho^A_U\ge\varrho^A_U$.
\item If $\theta\in\S^1$ is an eigenvalue of $U$ and
$\widetilde\varrho^A_U(\theta)>0$, then $\varrho^A_U(\theta)=0$. Otherwise,
$\varrho^A_U(\theta)=\widetilde\varrho^A_U(\theta)$.
\end{enumerate}
\end{Lemma}

One says that $A$ is conjugate to $U$ (or that $U$ satisfies a Mourre estimate) at a
point $\theta\in\S^1$ if $\widetilde\varrho^A_U(\theta)>0$, and that $A$ is strictly
conjugate to $U$ (or that $U$ satisfies a strict Mourre estimate) at $\theta$ if
$\varrho^A_U(\theta)>0$. Since $\widetilde\varrho^A_U(\theta)\ge\varrho^A_U(\theta)$
for each $\theta\in\S^1$ by Lemma \ref{lemma_properties}(c), strict conjugation is a
property stronger than conjugation. We write
\begin{equation}\label{def_mu}
\widetilde\mu^A(U):=\big\{\theta\in\S^1\mid\widetilde\varrho^A_U(\theta)>0\big\}
\end{equation}
for the subset of $\S^1$ where $A$ is conjugate to $U$, and note that
$\widetilde\mu^A(U)$ is open due to the lower semicontinuity of the function
$\widetilde\varrho^A_U$ (Lemma \ref{lemma_properties}(b)).

The next theorem corresponds to \cite[Thm.~3.4]{RST_2018}; its formulation has been
adapted to be consistent with the definition of locally smooth operators used in this
paper. We use the notation $\langle\cdot\rangle:=(1+|\cdot|^2)^{1/2}$ and use the
expression `spectral gap' to mean a hole in the spectrum.

\begin{Theorem}[Locally smooth operators]\label{thm_locally_smooth}
Let $U$ be a unitary operator in $\H$, let $A$ be a self-adjoint operator in $\H$, and
let $\G$ be an auxiliary Hilbert space. Assume either that $U$ has a spectral gap and
$U\in C^{1,1}(A)$, or that $U\in C^{1+0}(A)$. Suppose also there exist an open set
$\Theta\subset\S^1$, a number $a>0$ and an operator $K\in\K(\H)$ such that
$$
E^U(\Theta)\;\!U^{-1}[A,U]\;\!E^U(\Theta)\ge aE^U(\Theta)+K.
$$
Then each operator $T\in\B(\H,\G)$ which extends continuously to an element of
$\B\big(\dom(\langle A\rangle^s)^*,\G\big)$ for some $s>1/2$ is locally $U$-smooth on
any closed set $\Theta'\subset\Theta\setminus\sigma_{\rm p}(U)$.
\end{Theorem}

The last theorem of this section has been established in \cite[Thm.~2.7]{FRT13}:

\begin{Theorem}[Spectral properties]\label{thm_spec_prop}
Let $U$ be a unitary operator in $\H$ and let $A$ be a self-adjoint operator in $\H$.
Assume either that $U$ has a spectral gap and $U\in C^{1,1}(A)$, or that
$U\in C^{1+0}(A)$. Suppose also there exist an open set $\Theta\subset\S^1$, a number
$a>0$ and an operator $K\in\K(\H)$ such that
$$
E^U(\Theta)\;\!U^{-1}[A,U]\;\!E^U(\Theta)\ge aE^U(\Theta)+K.
$$
Then $U$ has at most finitely many eigenvalues in $\Theta$, each one of finite
multiplicity, and $U$ has no singular continuous spectrum in $\Theta$.
\end{Theorem}

\begin{Remark}\label{rem_abs}
(a) If the assumptions of Theorem \ref{thm_spec_prop} are satisfied with $K=0$, then
$U$ has no point spectrum in $\Theta$ either, meaning that the spectrum
of $U$ in $\Theta$ (if any) is purely absolutely continuous
(see \cite[Rem.~2.8]{FRT13}).

(b) An extension of the spectral result of Theorem \ref{thm_spec_prop} to unitary
operators $U\in C^{1,1}(A)$ without the gap assumption can be found in
\cite[Thm.~2.3]{ABC15_1}. It would be interesting to see if such an extension is also
possible for the result on locally $U$-smooth operators of Theorem
\ref{thm_locally_smooth} (the result of \cite[Prop.~7.1]{ABC15_1} in this direction
holds under an additional assumption which seems unnecessary).
\end{Remark}

\section{Mourre theory in two Hilbert spaces}\label{sec_two}
\setcounter{equation}{0}

From now on, in addition to the triple $(\H,U,A)$, we consider a second triple
$(\H_0,U_0,A_0)$ with $\H_0$ a Hilbert space, $U_0$ a unitary operator in $\H_0$, and
$A_0$ a self-adjoint operator in $\H_0$. We also asume that an identification operator
$J\in\B(\H_0,\H)$ is given.

The regularity of $U_0$ with respect to $A_0$ is usually easy to check, while the
regularity of $U$ with respect to $A$ is in general difficult to establish. In the
case of self-adjoint operators in one Hilbert space, various perturbative criteria
have been developed to tackle this problem, and often a distinction is made between
short-range and long-range case. For unitary operators, this distinction consists in
treating separately the two terms of the formal commutator $[A,U]=AU-UA$ in the
short-range case, or really computing the commutator $[A,U]$ in the long-range case.

In this section, we present an improved version of the results of
\cite[Sec.~3.2]{RST_2018} on the construction of a conjugate operator $A$ for $U$
starting from a conjugate operator $A_0$ for $U_0$ in the short-range case. Before
that, we recall a general result about the functions $\widetilde\varrho_U^A$ and
$\widetilde\varrho_{U_0}^{A_0}$ when the two-Hilbert space perturbation $V:=JU_0-UJ$
is compact:

\begin{Theorem}[Theorem 3.6 of \cite{RST_2018}]\label{thm_rho}
Let $(\H,U,A)$, $(\H_0,U_0,A_0)$ and $J,V\in\B(\H_0,\H)$ be as above, and assume that
\begin{enumerate}[(i)]
\item $U_0\in C^1(A_0)$ and $U\in C^1(A)$,
\item $JU_0^{-1}[A_0,U_0]J^*-U^{-1}[A,U]\in\K(\H)$,
\item $V\in\K(\H_0,\H)$,
\item For each $\eta\in C(\S^1,\R)$, one has $\eta(U)(JJ^*-1_\H)\eta(U)\in\K(\H)$.
\end{enumerate}
Then one has $\widetilde\varrho_U^A\ge \widetilde\varrho_{U_0}^{A_0}$.
\end{Theorem}

\subsection{Short-range case}\label{sec_short}

We start by showing how the condition $U\in C^1(A)$ and the assumptions (ii)-(iii) of
Theorem \ref{thm_rho} can be verified in the short-range case. Our approach consists
in deducing the desired information for $U$ from equivalent information on $U_0$. The
results are thus of a perturbative nature, and the price one has to pay is to impose
some compatibility conditions between $A_0$ and $A$.

Our first proposition is an improved version of \cite[Prop.~3.7]{RST_2018} that we
establish under less assumptions. We make use of the operator
$$
V_*:=JU_0^{-1}-U^{-1}J\in\B(\H_0,\H).
$$

\begin{Proposition}\label{prop_C1_short}
Let $U_0\in C^1(A_0)$, assume that $\DD\subset\H$ is a core for $A$ with
$J^*\DD\subset\dom(A_0)$, and suppose that
$\overline{V_*A_0\upharpoonright\dom(A_0)}\in\B(\H_0,\H)$ and
$\overline{(JA_0J^*-A)\upharpoonright\DD}\in\B(\H)$. Then $U\in C^1(A)$.
\end{Proposition}

\begin{proof}
First, we show that the assumption
$\overline{V_*A_0\upharpoonright\dom(A_0)}\in\B(\H_0,\H)$ implies that
$\overline{VA_0\upharpoonright\dom(A_0)}\in\B(\H_0,\H)$. Indeed, since
$U_0\in C^1(A_0)$, we have
$$
VA_0\upharpoonright\dom(A_0)
=-UV_*U_0A_0\upharpoonright\dom(A_0)
=-\big(UV_*A_0U_0+UV_*[U_0,A_0]\big)\upharpoonright\dom(A_0)
$$
with $[U_0,A_0]\in\B(\H_0,\H)$. Therefore, by using the inclusion
$U_0\dom(A_0)\subset\dom(A_0)$ and the assumption
$\overline{V_*A_0\upharpoonright\dom(A_0)}=B_*\in\B(\H_0,\H)$, we infer that
$$
\overline{VA_0\upharpoonright\dom(A_0)}
=-\big(UB_*U_0+UV_*[U_0,A_0]\big)
\in\B(\H_0,\H).
$$

Now, we show that $U\in C^1(A)$. For $\varphi\in\DD$, a direct calculation gives
\begin{align*}
\langle A\;\!\varphi,U\varphi\rangle_\H
-\langle\varphi,UA\;\!\varphi\rangle_\H
&=\langle A\;\!\varphi,U\varphi\rangle_\H
-\langle\varphi,UA\;\!\varphi\rangle_\H
-\langle\varphi,J\;\![A_0,U_0]J^*\varphi\rangle_\H
+\langle\varphi,J\;\![A_0,U_0]J^*\varphi\rangle_\H\\
&=\langle\varphi,VA_0J^*\varphi\rangle_\H
-\langle V_*A_0J^*\varphi,\varphi\rangle_\H
+\big\langle U^*\varphi,(JA_0J^*-A)\varphi\rangle_\H\\
&\quad-\langle(JA_0J^*-A)\varphi,U\varphi\rangle_\H
+\langle\varphi,J\;\![A_0,U_0]J^*\varphi\rangle_\H.
\end{align*}
Furthermore, we have
$$
\big|\langle\varphi,VA_0J^*\varphi\rangle_\H
-\langle V_*A_0J^*\varphi,\varphi\rangle_\H\big|
\le{\rm Const.}\;\!\|\varphi\|_\H^2
$$
due to the assumption $J^*\DD\subset\dom(A_0)$ and the inclusions
$\overline{VA_0\upharpoonright\dom(A_0)}\in\B(\H_0,\H)$ and
$\overline{V_*A_0\upharpoonright\dom(A_0)}\in\B(\H_0,\H)$, and we have
$$
\big|\langle U^*\varphi,(JA_0J^*-A)\varphi\rangle_\H
-\langle(JA_0J^*-A)\varphi,U\varphi\rangle_\H\big|
\le{\rm Const.}\;\!\|\varphi\|^2_\H
$$
due to the assumption $\overline{(JA_0J^*-A)\upharpoonright\DD}\in\B(\H)$. Finally,
since $U_0\in C^1(A_0)$ we also have
$$
\big|\langle\varphi,J\;\![A_0,U_0]J^*\varphi\rangle_\H\big|
\le{\rm Const.}\;\!\|\varphi\|^2_\H.
$$
In consequence, we obtain
$$
\big|\langle A\;\!\varphi,U\varphi\rangle_\H
-\langle\varphi,UA\;\!\varphi\rangle_\H\big|
\le{\rm Const.}\;\!\|\varphi\|^2_\H,
$$
which implies that $U\in C^1(A)$ due to the density of $\DD$ in $\dom(A)$.
\end{proof}

Our second proposition is an improved version of \cite[Prop.~3.8]{RST_2018} that we
establish under less assumptions. The proposition provides explicit conditions under
which assumption (ii) of Theorem \ref{thm_rho} is verified in the short-range case.

\begin{Proposition}\label{prop_com_compact}
Let $U_0\in C^1(A_0)$, assume that $\DD\subset\H$ is a core for $A$ with
$J^*\DD\subset\dom(A_0)$, and suppose that
$\overline{V_*A_0\upharpoonright\dom(A_0)}\in\K(\H_0,\H)$ and
$\overline{(JA_0J^*-A)\upharpoonright\DD}\in\K(\H)$. Then $U\in C^1(A)$ and
$$
JU_0^{-1}[A_0,U_0]J^*-U^{-1}[A,U]\in\K(\H).
$$
\end{Proposition}

\begin{proof}
First, we note that $U\in C^1(A)$ due to Proposition \ref{prop_C1_short}. So $[A,U]$
is a well-defined bounded operator. Next, the facts that $U_0\in C^1(A_0)$ and
$J^*\DD\subset\dom(A_0)$ imply the inclusions
$$
U_0J^*\DD\subset U_0\;\!\dom(A_0)\subset\dom(A_0).
$$
Using this and the assumptions
$\overline{V_*A_0\upharpoonright\dom(A_0)}\in\K(\H_0,\H)$ and
$\overline{(JA_0J^*-A)\upharpoonright\DD}\in\K(\H)$, we obtain for $\varphi\in\DD$ and
$\psi\in U^{-1}\DD$ that
\begin{align*}
&\big\langle\psi,\big(JU_0^{-1}[A_0,U_0]J^*-U^{-1}[A,U]\big)\varphi\big\rangle_\H\\
&=\langle\psi,V_*A_0U_0J^*\varphi\rangle_\H
+\langle V_*A_0J^*U\psi,\varphi\rangle_\H
+\langle(JA_0J^*-A)U\psi,U\varphi\rangle_\H
-\langle\psi,(JA_0J^*-A)\varphi\rangle_\H\\
&=\langle\psi,K_1U_0J^*\varphi\rangle_\H
+\langle K_1J^*U\psi,\varphi\rangle_\H
+\langle K_2U\psi,U\varphi\rangle_\H
-\langle\psi,K_2\varphi\rangle_\H
\end{align*}
with $K_1\in\K(\H_0,\H)$ and $K_2\in\K(\H)$. Since $\DD$ and $U^{-1}\DD$ are dense in
$\H$, it follows that $JU_0^{-1}[A_0,U_0]J^*-U^{-1}[A,U]\in\K(\H)$.
\end{proof}

In the rest of the section, we particularise the previous results to the case
$A=JA_0J^*$. This case deserves special attention since it is the most natural choice
of conjugate operator $A$ for $U$ when a conjugate operator $A_0$ for $U_0$ is given.
However, one needs in this case an assumption that guarantees the self-adjointness of
the operator $A:$

\begin{Assumption}\label{ass_eaa}
There exists a set $\DD\subset\dom(A_0J^*)\subset \H$ such that
$JA_0J^*\upharpoonright\DD$ is essentially self-adjoint, with closure denoted by $A$.
\end{Assumption}

Assumption \ref{ass_eaa} might be difficult to check in general, but in concrete
situations the choice of the set $\DD$ can be quite natural (see for example
\cite[Lemma~4.9]{RST_2018} for the case of quantum walks on $\Z$,
\cite[Rem.~4.3]{RT13_1} for the case of manifolds with asymptotically cylindrical
ends, or Lemma \ref{lemma_A} for the case of quantum walks on homogeneous trees).
Furthermore, if the operator $J$ is unitary, then Assumption \ref{ass_eaa} is
automatically satisfied with the set $\DD:=J\;\!\dom(A_0)$.

The two corollaries below follow directly from Propositions \ref{prop_C1_short} \&
\ref{prop_com_compact} when Assumption \ref{ass_eaa} is satisfied. They generalise
corollaries 3.10 \& 3.11 of \cite{RST_2018}.

\begin{Corollary}\label{Corol_C1(A)}
Let $U_0\in C^1(A_0)$, suppose that Assumption \ref{ass_eaa} holds, and assume that
$\overline{V_*A_0\upharpoonright\dom(A_0)}\in\B(\H_0,\H)$. Then $U\in C^1(A)$.
\end{Corollary}

\begin{Corollary}\label{Corol_est_supp}
Let $U_0\in C^1(A_0)$, suppose that Assumption \ref{ass_eaa} holds, and assume that
$\overline{V_*A_0\upharpoonright\dom(A_0)}\in\K(\H_0,\H)$. Then $U\in C^1(A)$ and
$$
JU_0^{-1}[A_0,U_0]J^*-U^{-1}[A,U]\in\K(\H).
$$
\end{Corollary}

\begin{Remark}\label{rem_enough}
(a) If needed, the assumption
$\overline{V_*A_0\upharpoonright\dom(A_0)}\in\B(\H_0,\H)$ in Proposition
\ref{prop_C1_short} and Corollary \ref{Corol_C1(A)} can be replaced by the assumption
$\overline{VA_0\upharpoonright\dom(A_0)}\in\B(\H_0,\H)$. Indeed, since
$U_0\in C^1(A_0)$, we have $U_0^{-1}\in C^1(A_0)$ and
$$
V_*A_0\upharpoonright\dom(A_0)
=-U^{-1}VU_0^{-1}A_0\upharpoonright\dom(A_0)
=-\big(U^{-1}VA_0U_0^{-1}+U^{-1}V[U_0^{-1},A_0]\big)\upharpoonright\dom(A_0)
$$
with $[U_0^{-1},A_0]\in\B(\H_0,\H)$. Therefore, by using the inclusion
$U_0^{-1}\dom(A_0)\subset\dom(A_0)$ and the assumption
$\overline{VA_0\upharpoonright\dom(A_0)}=B\in\B(\H_0,\H)$, we infer that
$$
\overline{V_*A_0\upharpoonright\dom(A_0)}
=-\big(U^{-1}BU_0^{-1}+U^{-1}V[U_0^{-1},A_0]\big)
\in\B(\H_0,\H).
$$

(b) Similarly, if $V\in\K(\H_0,\H)$, then the assumption
$\overline{V_*A_0\upharpoonright\dom(A_0)}\in\K(\H_0,\H)$ in Proposition
\ref{prop_com_compact} and Corollary \ref{Corol_est_supp} can be replaced by the
assumption $\overline{VA_0\upharpoonright\dom(A_0)}\in\K(\H_0,\H)$. Indeed, if
$V\in\K(\H_0,\H)$ and $\overline{VA_0\upharpoonright\dom(A_0)}=B\in\K(\H_0,\H)$, then
we obtain that
$$
\overline{V_*A_0\upharpoonright\dom(A_0)}
=-\big(U^{-1}BU_0^{-1}+U^{-1}V[U_0^{-1},A_0]\big)
\in\K(\H_0,\H).
$$
\end{Remark}

By combining the results of Theorem \ref{thm_rho}, Corollary \ref{Corol_est_supp}, and
Remark \ref{rem_enough}(b), we obtain a more explicit version of Theorem \ref{thm_rho}
in the case $A=JA_0J^*:$

\begin{Theorem}\label{thm_rho_bis}
Let $U_0,U$ be unitary operators in Hilbert spaces $\H_0,\H$, let $A_0$ be a
self-ajoint operator in $\H_0$, let $J\in\B(\H_0,\H)$, and assume that
\begin{enumerate}[(i)]
\item there exists a set $\DD\subset\dom(A_0J^*)\subset \H$ such that
$JA_0J^*\upharpoonright\DD$ is essentially self-adjoint, with closure denoted by $A$,
\item $U_0\in C^1(A_0)$,
\item $V\in\K(\H_0,\H)$,
\item $\overline{VA_0\upharpoonright\dom(A_0)}\in\K(\H_0,\H)$
\item for each $\eta\in C(\S^1,\R)$, one has $\eta(U)(JJ^*-1_\H)\eta(U)\in\K(\H)$.
\end{enumerate}
Then $U\in C^1(A)$ and $\widetilde\varrho_U^A\ge \widetilde\varrho_{U_0}^{A_0}$.
\end{Theorem}

\section{Scattering theory in two Hilbert spaces}\label{sec_scatt}
\setcounter{equation}{0}

In this section, we collect results on the existence and completeness under smooth
perturbations of the local wave operators for the triple $(U,U_0,J)$. Namely, given a
Borel set $\Theta\subset\S^1$, we present criteria for the existence and the
completeness of the strong limits
$$
W_\pm(U,U_0,J,\Theta):=\slim_{n\to\pm\infty}U^{-n}JU_0^nE^{U_0}(\Theta)
$$
under the assumption that $V=JU_0-UJ$ factorises as a product of locally smooth
operators on $\Theta$. We start by recalling the result of \cite[Lemma~2.1]{RST_2019}
on the intertwining property of wave operators (in \cite{RST_2019} the set $\Theta$ is
assumed to be open, but the proof holds for Borel sets too).

\begin{Lemma}[Intertwining property]\label{lemma_intertwinning}
Let $\Theta\subset\S^1$ be a Borel set and assume that $W_\pm(U,U_0,J,\Theta)$ exist.
Then $W_\pm(U,U_0,J,\Theta)$ satisfy for each bounded Borel function $\eta:\S^1\to\C$
the intertwining property
$$
W_\pm(U,U_0,J,\Theta)\;\!\eta(U_0)=\eta(U)\;\!W_\pm(U,U_0,J,\Theta).
$$
\end{Lemma}

Next, we define the closed subspaces of $\H$
$$
\NN_\pm(U,J,\Theta)
:=\Big\{\varphi\in\H\mid\lim_{n\to\pm\infty}
\big\|J^*U^nE^U(\Theta)\varphi\big\|_{\H_0}=0\Big\},
$$
and note that $E^U(\S^1\setminus\Theta)\H\subset\NN_\pm(U,J,\Theta)$, that $U$ is
reduced by $\NN_\pm(U,J,\Theta)$, and that
$$
\overline{\Ran\big(W_\pm(U,U_0,J,\Theta)\big)}\perp\NN_\pm(U,J,\Theta),
$$
this last fact being shown as in the self-adjoint case, see \cite[Lemma~3.2.1]{Yaf92}.
In particular, we have the inclusion
$$
\overline{\Ran\big(W_\pm(U,U_0,J,\Theta)\big)}
\subset E^U(\Theta)\H\ominus\NN_\pm(U,J,\Theta),
$$
which motivates the following definition:

\begin{Definition}[$J$-completeness]\label{def_J_complete}
Assume that $W_\pm(U,U_0,J,\Theta)$ exist. The operators $W_\pm(U,U_0,J,\Theta)$ are
$J$-complete on $\Theta$ if
$$
\overline{\Ran\big(W_\pm(U,U_0,J,\Theta)\big)}
=E^U(\Theta)\H\ominus\NN_\pm(U,J,\Theta).
$$
\end{Definition}

\begin{Remark}\label{rem_unitary}
If $J$ is unitary (as for instance when $\H_0=\H$ and $J=1_\H$), then
$\NN_\pm(U,J,\Theta)=E^U(\S^1\setminus\Theta)\H$, and the wave operators
$W_\pm(U,U_0,J,\Theta)$ have closed ranges because they are partial isometries with
initial sets $E^{U_0}(\Theta)\H_0$. Therefore, the $J$-completeness on $\Theta$
reduces to the completeness on $\Theta$ in the usual sense, namely,
$\Ran\big(W_\pm(U,U_0,J,\Theta)\big)=E^U(\Theta)\H$. In such a case, the wave
operators $W_\pm(U,U_0,J,\Theta)$ are unitary from $E^{U_0}(\Theta)\H_0$ to
$E^U(\Theta)\H$.
\end{Remark}

The following criterion for $J$-completeness has been shown in
\cite[Lemma~2.4]{RST_2019} for open sets $\Theta$, but the proof holds for Borel sets
too:

\begin{Lemma}\label{lemma_J_complete}
If $W_\pm(U,U_0,J,\Theta)$ and $W_\pm(U_0,U,J^*,\Theta)$ exist, then
$W_\pm(U,U_0,J,\Theta)$ are $J$-complete on $\Theta$.
\end{Lemma}

The next theorem corresponds to \cite[Thm.~2.5]{RST_2019}; its formulation has been
adapted to be consistent with the definition of locally smooth operators used in this
paper.

\begin{Theorem}[$J$-completeness of the wave operators]\label{thm_wave}
Let $\Theta\subset\S^1$ be an open set and let $\G$ be an auxiliary Hilbert space. If
$V=T^*T_0$ with $T_0\in\B(\H_0,\G)$ locally $U_0$-smooth on any closed set
$\Theta'\subset\Theta$ and $T\in\B(\H,\G)$ locally $U$-smooth on any closed set
$\Theta'\subset\Theta$, then the wave operators $W_\pm(U,U_0,J,\Theta)$ exist, are
$J$-complete on $\Theta$, and satisfy the relations
$$
W_\pm(U,U_0,J,\Theta)^*=W_\pm(U_0,U,J^*,\Theta)
\quad\hbox{and}\quad
W_\pm(U,U_0,J,\Theta)\eta(U_0)=\eta(U)W_\pm(U,U_0,J,\Theta)
$$
for each bounded Borel function $\eta:\S^1\to\C$.
\end{Theorem}

Now, by combining Theorem \ref{thm_locally_smooth} and Theorem \ref{thm_wave}, we
obtain the following criterion for the existence and $J$-completeness of the local
wave operators. We recall that the sets $\widetilde\mu^{A}(U)$,
$\widetilde\mu^{A_0}(U_0)$ have been defined in \eqref{def_mu}.

\begin{Corollary}\label{corol_wave}
Let $A_0,A$ self-adjoint operators in $\H_0,\H$. Assume either that $U_0$ and $U$ have
a spectral gap and $U_0\in C^{1,1}(A_0),U\in C^{1,1}(A)$, or that
$U_0\in C^{1+0}(A_0),U\in C^{1+0}(A)$. Let
$$
\Theta
:=\big\{\widetilde\mu^{A_0}(U_0)\cap\widetilde\mu^{A}(U)\big\}
\setminus\big\{\sigma_{\rm p}(U_0)\cup\sigma_{\rm p}(U)\big\},
$$
let $\G$ be an auxiliary Hilbert space, and suppose that $V=T^*T_0$ with
$T_0\in\B(\H_0,\G)$ extending continuously to an element of
$\B\big(\dom(\langle A_0\rangle^s)^*,\G\big)$ and $T\in\B(\H,\G)$ extending
continuously to an element of $\B\big(\dom(\langle A\rangle^s)^*,\G\big)$ for some
$s>1/2$. Then the wave operators $W_\pm(U,U_0,J,\Theta)$ exist, are $J$-complete on
$\Theta$, and satisfy the relations
$$
W_\pm(U,U_0,J,\Theta)^*=W_\pm(U_0,U,J^*,\Theta)
\quad\hbox{and}\quad
W_\pm(U,U_0,J,\Theta)\;\!\eta(U_0)=\eta(U)\;\!W_\pm(U,U_0,J,\Theta)
$$
for each bounded Borel function $\eta:\S^1\to\C$.
\end{Corollary}

In the last part of the section, we determine conditions under which the (global) wave
operators
$$
W_\pm(U,U_0,J):=\slim_{n\to\pm\infty}U^{-n}JU_0^nP_{\rm ac}(U_0)
$$
are complete in the usual sense; that is, satisfy
$\Ran\big(W_\pm(U,U_0,J)\big)=\H_{\rm ac}(U)$. When defining the completeness of the
wave operators $W_\pm(U,U_0,J)$, the simplest choice would be to require in addition
that $W_\pm(U,U_0,J)$ are partial isometries with initial subspaces
$\H_0^\pm=\H_{\rm ac}(U_0)$ (as in \cite[Def.~III.9.24]{BW83} or
\cite[Def.~2.3.1]{Yaf92} in the self-adjoint case). However, in applications it may
happen that the ranges of $W_\pm(U,U_0,J)$ are equal to $\H_{\rm ac}(U)$ but that
$\H_0^\pm\neq\H_{\rm ac}(U_0)$. In the self-ajoint case, this typically happens for
multichannel type scattering processes, in which case the usual criteria for
completeness, as \cite[Prop.~III.9.40]{BW83} or \cite[Thm.~2.3.6]{Yaf92}, do not be
apply. This is why we present below a criterion fot the completeness of the wave
operators $W_\pm(U,U_0,J)$ without assuming that $\H_0^\pm=\H_{\rm ac}(U_0)$.

We start by establishing a chain rule for wave operators analogous to
\cite[Thm.~2.1.7]{Yaf92} in the self-adjoint case.

\begin{Lemma}[Chain rule]\label{lemma_chain}
Let $U_1,U_2,U_3$ be unitary operators in Hilbert spaces $\H_1,\H_2,\H_3$, let
$J_{23}\in\B(\H_2,\H_3)$ and $J_{31}\in\B(\H_3,\H_1)$, and assume that the strong
limits
$$
W_\pm(U_3,U_2,J_{23}):=\slim_{n\to\pm\infty}U_3^{-n}J_{23}U_2^nP_{\rm ac}(U_2)
\quad\hbox{and}\quad
W_\pm(U_1,U_3,J_{31}):=\slim_{n\to\pm\infty}U_1^{-n}J_{31}U_3^nP_{\rm ac}(U_3)
$$
exist. Then the strong limits
$$
W_\pm(U_1,U_2,J_{31}J_{23})
:=\slim_{n\to\pm\infty}U_1^{-n}J_{31}J_{23}U_2^nP_{\rm ac}(U_2)
$$
exist and satisfy the chain rule
$$
W_\pm(U_1,U_2,J_{31}J_{23})=W_\pm(U_1,U_3,J_{31})\;\!W_\pm(U_3,U_2,J_{23}).
$$
\end{Lemma}

\begin{proof}
For any $n\in\Z$, we have the identity
\begin{equation}\label{eq_rhs}
U_1^{-n}J_{31}J_{23}U_2^nP_{\rm ac}(U_2)
=U_1^{-n}J_{31}U_3^n\big(P_{\rm ac}(U_3)+1-P_{\rm ac}(U_3)\big)
U_3^{-n}J_{23}U_2^nP_{\rm ac}(U_2).
\end{equation}
But due to the fact that $\big(1-P_{\rm ac}(U_3)\big)W_\pm(U_3,U_2,J_{23})=0$ (which
can be shown as in the self-adjoint case, see \cite[Eq.~2.1.18]{Yaf92}), we have for
any $\varphi_2\in\H_2$ that
$$
\lim_{n\to\pm\infty}\big\|\big(1-P_{\rm ac}(U_3)\big)U_3^{-n}J_{23}U_2^n
P_{\rm ac}(U_2)\varphi_2\big\|_{\H_3}=0.
$$
Therefore, we obtain that
\begin{align*}
\slim_{n\to\pm\infty}U_1^{-n}J_{31}J_{23}U_2^nP_{\rm ac}(U_2)
&=\slim_{n\to\pm\infty}U_1^{-n}J_{31}U_3^nP_{\rm ac}(U_3)
U_3^{-n}J_{23}U_2^nP_{\rm ac}(U_2)\\
&=\slim_{n\to\pm\infty}U_1^{-n}J_{31}U_3^nP_{\rm ac}(U_3)
\cdot\slim_{n\to\pm\infty}U_3^{-n}J_{23}U_2^nP_{\rm ac}(U_2)\\
&=W_\pm(U_1,U_3,J_{31})\;\!W_\pm(U_3,U_2,J_{23}),
\end{align*}
which proves the claims.
\end{proof}

We can now present our theorem on the completeness of the wave operators
$W_\pm(U,U_0,J)$. It is a unitary analogue of \cite[Prop.~5.1]{RT13_2} in the
self-adjoint case.

\begin{Theorem}[Completeness of the wave operators]\label{thm_wave_bis}
Suppose that the wave operators $W_\pm(U,U_0,J)$ exist and are partial isometries with
initial set projections $P_0^\pm$. If there is $J'\in\B(\H,\H_0)$ such that
\begin{equation}\label{eq_wave_prime}
W_\pm(U_0,U,J')
:=\slim_{n\to\pm\infty}U_0^{-n}J'U^nP_{\rm ac}(U)
\end{equation}
exist and
\begin{equation}\label{eq_JJ_prime}
\slim_{n\to\pm\infty}(JJ'-1)U^nP_{\rm ac}(U)=0,
\end{equation}
then $\Ran\big(W_\pm(U,U_0,J)\big)=\H_{\rm ac}(U)$. Conversely, if
$\Ran\big(W_\pm(U,U_0,J)\big)=\H_{\rm ac}(U)$ and there is $J'\in\B(\H,\H_0)$ such
that
\begin{equation}\label{eq_prime_JJ}
\slim_{n\to\pm\infty}(J'J-1)U_0^nP_0^\pm=0,
\end{equation}
then $W_\pm(U_0,U,J')$ exist and \eqref{eq_JJ_prime} holds.
\end{Theorem}

\begin{proof}
(i) An application of Lemma \ref{lemma_chain} with $U_1=U_2=U$, $U_3=U_0$,
$\H_1=\H_2=\H$, $\H_3=\H_0$, $J_{23}=J'$, and $J_{31}=J$, implies that the strong
limits
$$
W_\pm (U,U,JJ')
:=\slim_{n\to\pm\infty}U^{-n}JJ'U^nP_{\rm ac}(U)
$$
exist and satisfy the chain rule
\begin{equation}\label{eq_chain}
W_\pm(U,U,JJ')=W_\pm(U,U_0,J)\;\!W_\pm(U_0,U,J').
\end{equation}
In consequence, the equality
$$
\slim_{n\to\pm\infty}\big(U^{-n}JJ'U^nP_{\rm ac}(U)-P_{\rm ac}(U)\big)=0,
$$
which follow from \eqref{eq_JJ_prime}, implies that
$W_\pm(U,U,JJ')P_{\rm ac}(U)=P_{\rm ac}(U)$. This, together with \eqref{eq_chain} and
the equality $W_\pm(U_0,U,J')=W_\pm(U_0,U,J')P_{\rm ac}(U)$,
gives
$$
W_\pm(U,U_0,J)W_\pm(U_0,U,J')
=W_\pm(U,U,JJ')P_{\rm ac}(U)
=P_{\rm ac}(U),
$$
which is equivalent to
$$
W_\pm(U_0,U,J')^*W_\pm(U,U_0,J)^*=P_{\rm ac}(U).
$$
This implies the inclusion
$
\ker\big(W_\pm(U,U_0,J)^*\big)\subset\H_{\rm ac}(U)^\perp
$,
which together with the fact that the range of a partial isometry is closed leads to
the inclusion
$$
\H=\Ran\big(W_\pm(U,U_0,J)\big)\oplus\ker\big(W_\pm(U,U_0,J)^*\big)
\subset\H_{\rm ac}(U)\oplus\H_{\rm ac}(U)^\perp
=\H.
$$
So, one must have $\Ran\big(W_\pm(U,U_0,J)\big)=\H_{\rm ac}(U)$, which proves the
first claim.

(ii) For the converse, let $\varphi\in\H_{\rm ac}(U)$. Then we know from the
assumption $\Ran\big(W_\pm(U,U_0,J)\big)=\H_{\rm ac}(U)$ that there exist
$\varphi_0^\pm\in P_0^\pm\H_0$ such that
\begin{equation}\label{eq_converse}
\lim_{n\to\pm\infty}\big\|U^n\varphi-JU_0^nP_0^\pm\varphi_0^\pm\big\|_\H=0.
\end{equation}
Together with \eqref{eq_prime_JJ}, this implies that the norm
\begin{align*}
\big\|U_0^{-n}J'U^n\varphi-P_0^\pm\varphi_0^\pm \big\|_{\H_0}
&\le\big\|U_0^{-n}J'
\big(U^n\varphi-JU_0^nP_0^\pm\varphi_0^\pm\big)\big\|_{\H_0}
+\big\|U_0^{-n}J'JU_0^nP_0^\pm\varphi_0^\pm
-P_0^\pm\varphi_0^\pm\big\|_{\H_0}\\
&\le{\rm Const.}\;\!\big\|U^n\varphi-JU_0^nP_0^\pm\varphi_0^\pm\big\|_\H
+\big\|(J'J-1)U_0^nP_0^\pm\varphi_0^\pm \big\|_{\H_0}
\end{align*}
converges to $0$ as $n\to\pm\infty$, showing that the wave operators
\eqref{eq_wave_prime} exist.

To show \eqref{eq_JJ_prime}, we first observe that \eqref{eq_prime_JJ} gives
$$
\slim_{n\to\pm\infty}(JJ'-1)JU_0^nP_0^\pm
=\slim_{n\to\pm\infty}J(J'J-1)U_0^nP_0^\pm
=0.
$$
Together with \eqref{eq_converse}, this implies that the norm
\begin{align*}
\big\|(JJ'-1)U^n\varphi\big\|_\H
&\le\big\|(JJ'-1)\big(U^n\varphi-JU_0^nP_0^\pm\varphi_0^\pm\big)\big\|_\H
+\big\|(JJ'-1)JU_0^nP_0^\pm\varphi_0^\pm\big\|_\H\\
&\le{\rm Const.}\;\!\big\|JU_0^nP_0^\pm\varphi_0^\pm-U^n\varphi\big\|_\H
+\big\|(JJ'-1)JU_0^nP_0^\pm\varphi_0^\pm\big\|_\H
\end{align*}
converges to $0$ as $n\to\pm\infty$, showing \eqref{eq_JJ_prime}.
\end{proof}

\section{Quantum walks on homogenous trees of odd degree}\label{section_tree}
\setcounter{equation}{0}

In this section, we use the theory of Sections \ref{sec_one}-\ref{sec_scatt} to
determine spectral and scattering properties of a class of anisotropic quantum walks
on homogenous trees of odd degree. We start by recalling from \cite{HJ_2014} the
necessary material for the definition of the quantum walks.

Let $\Tau$ be an homogenous tree of odd degree $d\ge3$, that is, a finitely generated
group $\Tau$ with generators $a_1,\dots,a_d$, neutral element $e$, and presentation
$$
\Tau:=\big\langle a_1,\dots,a_d\mid a_1^2=\dots=a_d^2=e\big\rangle.
$$
See Figure \ref{fig_tree} for an illustration in the case $d=3$, with progress away
from $e$ corresponding to multiplication from the right. Using the word length
$|\cdot|:\Tau\to\N$ \cite[Sec.~6.1]{CC_2010}, we can define for each $i=1,\dots,d$ the
$i$-th main branch of $\Tau$, that is, the subtree
$$
\Tau_i:=\big\{x\in\Tau\mid|a_ix|=|x|-1\big\},
$$
and we can define the sets $\Tau_{\rm e}$ and $\Tau_{\rm o}$ of even or odd elements
of $\Tau$
$$
\Tau_{\rm e}:=\big\{x\in\Tau\mid|x|\in2\N\big\}
\quad\hbox{and}\quad
\Tau_{\rm o}:=\big\{x\in\Tau\mid|x|\in2\N+1\big\}.
$$
We write $\chi_{\mathcal B}$ for the characteristic function of a set
${\mathcal B}\subset\Tau$. So, $\chi_{\Tau_i}$, $\chi_{\Tau_{\rm e}}$,
$\chi_{\Tau_{\rm o}}$ denote the characteristic functions of $\Tau_i$, $\Tau_{\rm e}$,
$\Tau_{\rm o}$, respectively. We also use the shorthand notations
$$
\chi_1:=\chi_{\Tau_1\cup\{e\}},\quad
\chi_i:=\chi_{\Tau_i}~(i\ge2),\quad
\chi_{\rm e}:=\chi_{\Tau_{\rm e}},\quad
\chi_{\rm o}:=\chi_{\Tau_{\rm o}},
$$
which allow to write succinctly two partitions of unity of $\Tau:$
$$
\sum_{k=1}^d\chi_k\equiv1
\quad\hbox{and}\quad
\chi_{\rm e}+\chi_{\rm o}\equiv1.
$$
By letting $\ell^2(\Tau)$ be the Hilbert space of square-summable functions $\Tau\to\C$
with scalar product
$$
\langle f,g\rangle_{\ell^2(\Tau)}:=\sum_{x\in\Tau}f(x)\overline{g(x)},
\quad f,g\in\ell^2(\Tau),
$$
and $\delta_x$ be the element of the canonical basis of $\ell^2(\Tau)$ which sits at
$x\in\Tau$ (i.e. $\delta_x(y):=\delta_{x,y}$ for $y\in\Tau$), we can decompose
$\chi_1$ as $\chi_1=\chi_{\Tau_1}+\delta_e$.

For $i,j=1,\ldots,d$, we define the right translation $R_i\in\B\big(\ell^2(\Tau)\big)$
and the shift $S_{i,j}\in\B\big(\ell^2(\Tau)\big)$ by
$$
R_if:=f(\;\!\cdot\;\!a_i),\quad f\in\ell^2(\Tau),
$$
and
$$
S_{i,j}:=\chi_{\rm e}R_i+\chi_{\rm o}R_j.
$$
The operators $R_i$ and $S_{i,j}$ are unitary and satisfy the relations
$R_i^{-1}=R_i=S_{i,i}$ and $S_{i,j}^{-1}=S_{j,i}$. Let $\H:=\ell^2(\Tau,\C^d)$
the Hilbert space of square-summable functions $\Tau\to\C^d$ with scalar product
$$
\langle\varphi,\psi\rangle_\H:=\sum_{x\in\Tau}\langle\varphi(x),\psi(x)\rangle_{\C^d},
\quad \varphi,\psi\in\H.
$$
Then the evolution operator of the quantum walk that we consider is the product
$U:=SC$ in $\H$, where $S$ is the (diagonal unitary) shift operator given by
$$
S:=
\left(\begin{smallmatrix}
S_{1+1,1+2} &&&\\
& S_{2+1,2+2} && \mbox{\large$0$}\\
\mbox{\large$0$} && \ddots &\\
&&& S_{d+1,d+2}\\
\end{smallmatrix}\right)
\quad\hbox{with}\quad S_{d,d+1}:=S_{d,1},~S_{d+1,d+2}:=S_{1,2},
$$
and $C$ the (unitary) coin operator given by
$$
(C\varphi)(x):=C(x)\varphi(x),\quad\varphi\in\H,~x\in\Tau,~C(x)\in\U(d).
$$
We assume that the coin operator $C$ has an anisotropic behaviour at infinity; it
converges with short-range rate to $d$ asymptotic coin operators, one on each main
branch of $\Tau$, in the following way:

\begin{Assumption}[Short-range]\label{ass_short}
For each $i=1,\ldots,d$, there exist a diagonal matrix $C_i\in\U(d)$ and a scalar
$\varepsilon_i>0$ such that
$$
\big\|C(x)-C_i\big\|_{\B(\C^d)}\le{\rm Const.}\;\!\langle x\rangle^{-(1+\varepsilon_i)}
\quad\hbox{if $x\in\Tau_i$.}
$$
\end{Assumption}

This assumption provides $d$ new evolution operators $U_i:=SC_i$ describing the
asymptotic behaviour of $U$ on each main branch $\Tau_i$. It also suggests to define
the free evolution operator as the direct sum operator
$$
U_0:=\bigoplus_{k=1}^dU_k\quad\hbox{in}\quad\H_0:=\bigoplus_{k=1}^d\H,
$$
and to define the identification operator $J\in\B(\H_0,\H)$ as
$$
J\;\!\Phi:=\sum_{k=1}^d\chi_k\;\!\varphi_k,
\quad\Phi=(\varphi_1,\dots,\varphi_d)\in\H_0.
$$

Using the same notation for functions and the corresponding multiplication operators
in $\H$, we directly get:

\begin{Lemma}\label{lemma_J}
The adjoint $J^*\in\B(\H,\H_0)$ is given by
$$
J^*\varphi=\big(\chi_1\;\!\varphi,\dots,\chi_d\;\!\varphi\big),
\quad\varphi\in\H,
$$
and satisfies the relations $J^*J=\bigoplus_{k=1}^d\chi_k$ and $JJ^*=1_\H$.
\end{Lemma}

\begin{Remark}
Choosing $U_0$ as a free evolution operator for $U$ is not the only possibility.
Like for other quantum systems with discrete or continuous time variable, several
choices are possible for the free evolution operator. We will discuss some of them in
Section \ref{sec_wave}.
\end{Remark}

\subsection{Free evolution operator}\label{section_free}

In this section, we construct a conjugate operator for the free evolution operator
$U_0$ and determine the spectral properties of $U_0$. For this, one first needs to
construct a conjugate operator for the shifts $S_{i,j}$ ($i\ne j$). Following the
intuition provided by classical mechanics \cite[Sec.~1.1]{Tie_2013}, the approach we
use is to define the conjugate operator for $S_{i,j}$ as an operator of the form
$S_{i,j}^{-1}[X^2,S_{i,j}]$, where $X$ is an appropriate position observable in
$\ell^2(\Tau)$ growing along the discrete flow generated by $S_{i,j}$. A natural
candidate for $X$ is the operator of multiplication by the word length $|\cdot|$.
However, the fact that in general $|\cdot|$ does not take the value $0$ on the support
of the iterates $S_{i,j}^nf$ ($n\in\Z,f\in\ell^2(\Tau)$) generates technical
difficulties. See Figure \ref{fig_iterates} for an example where $d=3$, $(i,j)=(1,2)$,
$f=\delta_{a_3}$, and the minimum value taken by $|\cdot|$ on the support of the
iterates is $1$. To correct this issue, one has to use modifications of the word length
consisting in $|\cdot|$ composed with a translation making the support of the iterates
$S_{i,j}^nf$ pass by the neutral element (in Figure \ref{fig_iterates}, this would be
the translation by $a_3$). For each $i\ne j$, the modified word length is defined as
$$
|\cdot|_{i,j}:\Tau\to\N,~~x\mapsto|x|_{i,j}:=\big|x_{i,j}^{-1}x\big|,
$$
where $x_{i,j}$ is the longest word in the reduced representation of $x$, starting
from the left and ending with a letter different from $a_i$ or $a_j$. For instance, if
$x=a_1a_2a_3a_2$, then $x_{1,2}=a_1a_2a_3$, and if $x=a_1a_2a_1a_2$, then $x_{1,2}=e$.

\begin{figure}[h]
\centering
\includegraphics[width=245pt]{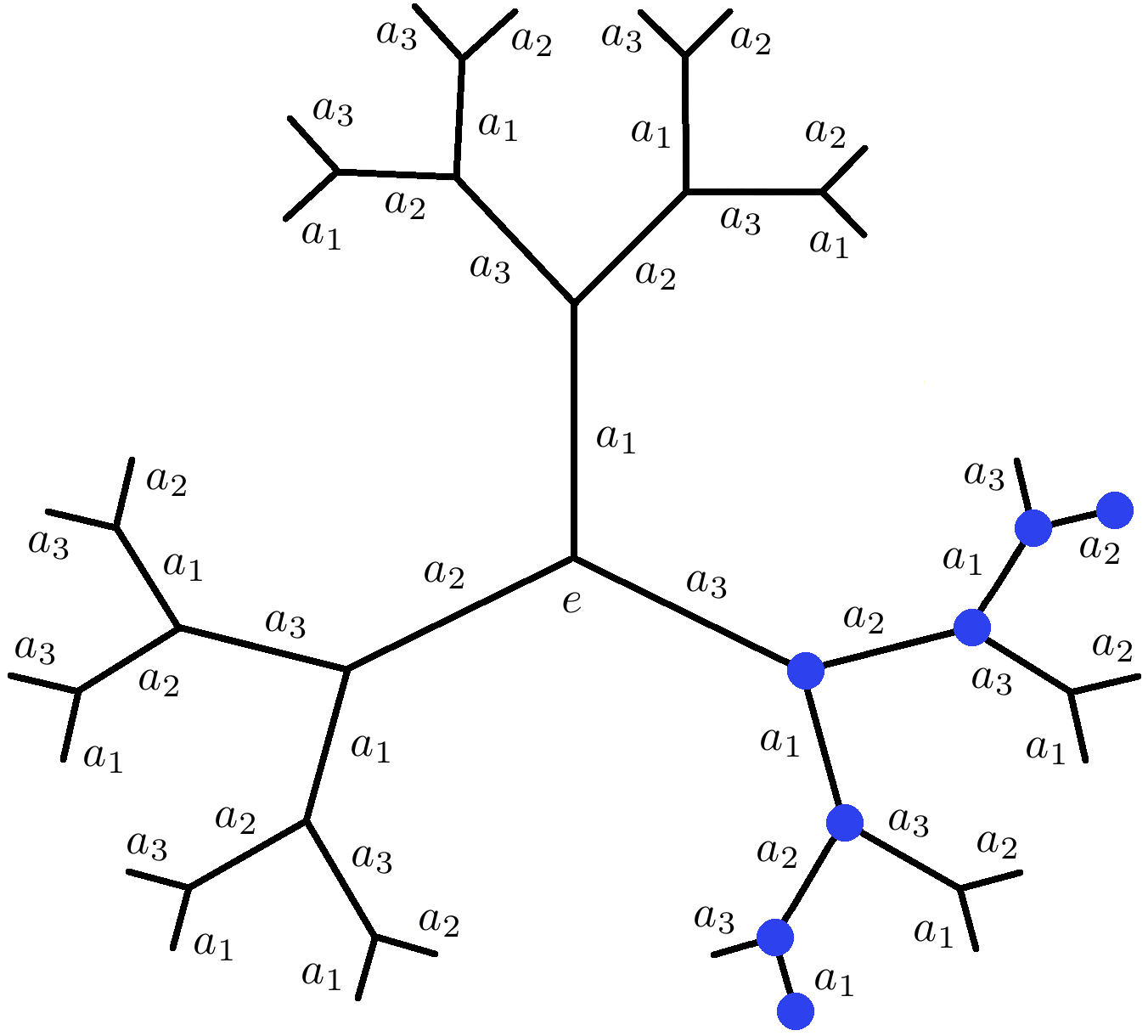}
\caption{Support of the iterates $S_{12}^n\delta_{a_3}$ ($n\in\Z$) when $d=3$}
\label{fig_iterates}
\end{figure}

Using these modified word lengths, we can now construct the conjugate operator for
$S_{i,j}$. We write $C_{\rm c}(\Tau)$ for the set of functions $\Tau\to\C$ with compact
support.

\begin{Lemma}[Conjugate operator for $S_{i,j}$]\label{lemma_SC_0_ij}
Take $i\ne j$. Then the operator
\begin{equation}\label{eq_A_ij}
A_{i,j}f
:=S_{i,j}^{-1}\big[|\cdot|_{i,j}^2,S_{i,j}\big]f
=\big(\chi_{\rm e}\;\!|\cdot a_j|_{i,j}^2+\chi_{\rm o}\;\!|\cdot a_i|_{i,j}^2
-|\cdot|_{i,j}^2\big)f,\quad f\in C_{\rm c}(\Tau),
\end{equation}
is essentially self-adjoint in $\ell^2(\Tau)$, with closure also denoted by $A_{i,j}$.
Furthermore, one has $S_{i,j}\in C^\infty(A_{i,j})$ with
$$
S_{i,j}^{-1}[A_{i,j},S_{i,j}]=2\cdot1_{\ell^2(\Tau)}.
$$
\end{Lemma}

\begin{proof}
(i) The equality
$$
S_{i,j}^{-1}\big[|\cdot|_{i,j}^2,S_{i,j}\big]f
=\big(\chi_{\rm e}\;\!|\cdot a_j|_{i,j}^2+\chi_{\rm o}\;\!|\cdot a_i|_{i,j}^2
-|\cdot|_{i,j}^2\big)f,\quad f\in C_{\rm c}(\Tau),
$$
follows from a direct calculation, and the essential self-adjointness of $A_{i,j}$
follows from the fact that multiplication operators are essentially self-adjoint on
the set of functions with compact support \cite[Ex.~5.1.15]{Ped_1989}. Now, further
calculations give
$$
S_{i,j}^{-1}[A_{i,j},S_{i,j}]f
=\big\{\chi_{\rm e}\big(|\cdot a_ja_i|_{i,j}^2-2|\cdot a_j|_{i,j}^2+|\cdot|_{i,j}^2\big)
+\chi_{\rm o}\big(|\cdot a_ia_j|_{i,j}^2-2|\cdot a_i|_{i,j}^2+|\cdot|_{i,j}^2\big)\big\}f,
$$
and in point (ii) below we show that
\begin{equation}\label{eq_equal_2}
\chi_{\rm e}(x)\big(|xa_ja_i|_{i,j}^2-2|xa_j|_{i,j}^2+|x|_{i,j}^2\big)
+\chi_{\rm o}(x)\big(|xa_ia_j|_{i,j}^2-2|xa_i|_{i,j}^2+|x|_{i,j}^2\big)=2
\quad\hbox{for all $x\in\Tau$.}
\end{equation}
Therefore, it follows that $S_{i,j}^{-1}[A_{i,j},S_{i,j}]=2\cdot1_{\ell^2(\Tau)}$ on
$C_{\rm c}(\Tau)$, and since $C_{\rm c}(\Tau)$ is a core for $A_{i,j}$ this implies
that $S_{i,j}\in C^\infty(A_{i,j})$ with
$S_{i,j}^{-1}[A_{i,j},S_{i,j}]=2\cdot1_{\ell^2(\Tau)}$.

(ii) To prove \eqref{eq_equal_2} it is sufficient to show that
$$
|xa_ja_i|_{i,j}^2-2|xa_j|_{i,j}^2+|x|_{i,j}^2=2\hbox{ if $x\in\Tau_{\rm e}$}
\quad\hbox{and}\quad
|xa_ia_j|_{i,j}^2-2|xa_i|_{i,j}^2+|x|_{i,j}^2=2\hbox{ if $x\in\Tau_{\rm o}$.}
$$
We only give the proof of the first identity, since the second one is similar. If
$|x|_{i,j}=0$, then $x_{i,j}^{-1}x=e$ and
$$
|xa_ja_i|_{i,j}^2-2|xa_j|_{i,j}^2+|x|_{i,j}^2
=|a_ja_i|^2-2|a_j|^2+0^2
=2.
$$
If $|x|_{i,j}=1$, then $x_{i,j}^{-1}x=a_i$ or $x_{i,j}^{-1}x=a_j$. In the first case, we get
$$
|xa_ja_i|_{i,j}^2-2|xa_j|_{i,j}^2+|x|_{i,j}^2
=|a_ia_ja_i|^2-2|a_ia_j|^2+1^2
=2,
$$
and in the second case we get
$$
|xa_ja_i|_{i,j}^2-2|xa_j|_{i,j}^2+|x|_{i,j}^2
=|a_i|^2-2|e|^2+1^2
=2.
$$
Finally, if $|x|_{i,j}\ge2$, then $|xa_j|_{i,j}=|x|_{i,j}+1$ or $|xa_j|_{i,j}=|x|_{i,j}-1$. In the first
case, we get $|xa_ja_i|_{i,j}=|x|_{i,j}+2$ and
$$
|xa_ja_i|_{i,j}^2-2|xa_j|_{i,j}^2+|x|_{i,j}^2
=(|x|_{i,j}+2)^2-2(|x|_{i,j}+1)^2+|x|_{i,j}^2
=2,
$$
and in the second case we get $|xa_ja_i|_{i,j}=|x|_{i,j}-2$ and
$$
|xa_ja_i|_{i,j}^2-2|xa_j|_{i,j}^2+|x|_{i,j}^2
=(|x|_{i,j}-2)^2-2(|x|_{i,j}-1)^2+|x|_{i,j}^2
=2.
$$
Note that we cannot have $|xa_ja_i|_{i,j}=|x|_{i,j}$ due to the very definition of $|\cdot|_{i,j}$.
\end{proof}

Using the operators $A_{i,j}$ we can construct a conjugate operator for operators
$\widetilde U_0:=SC_0$, where $C_0$ is a diagonal coin operator
$$
(C_0\varphi)(x):=C_0(x)\varphi(x),
\quad\hbox{$\varphi\in\H$, $x\in\Tau$, $C_0(x)\in\U(d)$ and diagonal.}
$$
We use the notation $C_{\rm c}(\Tau,\C^d)$ for the set of functions $\Tau\to\C^d$ with
compact support and we set $A_{d,d+1}:=A_{d,1}$ and $A_{d+1,d+2}:=A_{1,2}$.

\begin{Lemma}[Conjugate operator for $\widetilde U_0$]\label{lemma_SC_0}
The operator
$$
\widetilde A\;\!\varphi
:=\left(\begin{smallmatrix}
A_{1+1,1+2} &&&\\
& A_{2+1,2+2} && \mbox{\large$0$}\\
\mbox{\large$0$} && \ddots &\\
&&& A_{d+1,d+2}\\
\end{smallmatrix}\right)
\varphi,
\quad\varphi\in C_{\rm c}(\Tau,\C^d),
$$
is essentially self-adjoint in $\H$, with closure also denoted by $\widetilde A$.
Furthermore, one has $\widetilde U_0\in C^\infty(\widetilde A)$ with
$$
\widetilde U_0^{-1}[\widetilde A,\widetilde U_0]=2\cdot1_\H.
$$
\end{Lemma}

\begin{proof}
The essential self-adjointness of $\widetilde A$ on $C_{\rm c}(\Tau,\C^d)\subset\H$
follows from the isomorphism $\H\simeq\bigoplus_{k=1}^d\ell^2(\Tau)$ and the fact that
the operators $A_{i,j}$ are essentially self-adjoint on
$C_{\rm c}(\Tau)\subset\ell^2(\Tau)$. Using the same isomorphism, Lemma
\ref{lemma_SC_0_ij}, and the relation $[\widetilde A,C_0]=0$, one obtains that
$\widetilde U_0\in C^\infty(\widetilde A)$ with
$$
\widetilde U_0^{-1}[\widetilde A,\widetilde U_0]
=C_0^{-1}\big(S^{-1}[\widetilde A,S]\big)C_0
=C_0^{-1}(2\cdot1_\H)C_0
=2\cdot1_\H.
$$
\end{proof}

Lemma \ref{lemma_SC_0}, Theorem \ref{thm_spec_prop}, and Remark \ref{rem_abs}(a)
imply that $\widetilde U_0$ has purely absolutely continuous spectrum. But more can be
said. Since $\widetilde A$ is essentially self-adjoint on $C_{\rm c}(\Tau,\C^d)$, the
identity $\widetilde U_0^{-1}[\widetilde A,\widetilde U_0]=2\cdot1_\H$ implies that
$\widetilde U_0^{-1}\widetilde A\;\!\widetilde U_0=\widetilde A+2\cdot1_\H$. Using
this relation and functional calculus, we obtain for any $s\in\R$ and
$\gamma\in C(\S^1)$ that
\begin{equation}\label{eq_imp}
\e^{is\widetilde A}\gamma(\widetilde U_0)\e^{-is\widetilde A}
=\gamma\big(\widetilde U_0\e^{is\widetilde U_0^{-1}\widetilde A\;\!\widetilde U_0}
\e^{-is\widetilde A}\big)
=\gamma\big(\widetilde U_0\e^{is(\widetilde A+2\cdot1_\H)}\e^{-is\widetilde A}\big)
=\gamma(\e^{2is}\widetilde U_0).
\end{equation}
The relation
$
\e^{is\widetilde A}\gamma(\widetilde U_0)\e^{-is\widetilde A}
=\gamma(\e^{2is}\widetilde U_0)
$
and Mackey's imprimitivity theorem \cite[Thm.~5]{Ors_1979} applied to the group $\R$
and the subgroup $\Z$ imply the existence of a continuous unitary representation
$\sigma$ of $\Z$ in a Hilbert space $\mathfrak h_\sigma$ achieving the following: Let
$F_\sigma$ be the set of functions $f_\sigma:\R\to\mathfrak h_\sigma$ such that
\begin{enumerate}[(i)]
\item $f_\sigma(n+s)=\sigma(n)f_\sigma(s)$ for all $n\in\Z$ and $s\in\R$,
\item $\|f_\sigma(\;\!\cdot\;\!)\|_{\mathfrak h_\sigma}\in\ltwoloc(\R)$,
\item $f_\sigma$ is strongly measurable,
\end{enumerate}
let $\langle\cdot,\cdot\rangle_{\H_\sigma}$ and $\|\cdot\|_{\H_\sigma}$ be the scalar
product and norm on $F_\sigma$ given by
$$
\langle f_\sigma,g_\sigma\rangle_{\H_\sigma}
:=\int_0^1\d s\,\langle f_\sigma(s),g_\sigma(s)\rangle_{\mathfrak h_\sigma}
\quad\hbox{and}\quad
\|f_\sigma\|_{\H_\sigma}:=\sqrt{\langle f_\sigma,f_\sigma\rangle_{\H_\sigma}},
\quad f_\sigma,g_\sigma\in F_\sigma,
$$
and let $\H_\sigma$ be the Hilbert space completion of $F_\sigma$ for the norm
$\|\cdot\|_{\H_\sigma}$, that is,
$$
\H_\sigma:=\{f_\sigma\in F_\sigma\mid\|f_\sigma\|_{\H_\sigma}<\infty\}
/\{f_\sigma\in F_\sigma\mid\|f_\sigma\|_{\H_\sigma}=0\}.
$$
Then there exists a unitary operator $\UU:\H\to\H_\sigma$ satisfying for each
$s\in\R$ and $\gamma\in C(\S^1)$
$$
\UU\e^{\pi is\widetilde A}\UU^{-1}=U_\sigma(s)
\quad\hbox{and}\quad
\UU\gamma(\widetilde U_0)\UU^{-1}=P_\sigma(\gamma),
$$
with $U_\sigma$ the induced continuous unitary representation of $\sigma$ from $\Z$ to
$\R$ given by
$$
\big(U_\sigma(s)f_\sigma\big)(t):=f_\sigma(t+s),\quad s,t\in\R,~f_\sigma\in\H_\sigma,
$$
and $P_\sigma$ given by
$$
\big(P_\sigma(\gamma)f_\sigma\big)(s):=\gamma(\e^{2\pi is})f_\sigma(s),
\quad s\in\R,~f_\sigma\in\H_\sigma,~\gamma\in C(\S^1).
$$
Therefore, the operator $\widetilde U_0$ is unitarily equivalent to a multiplication
operator with purely absolutely continuous spectrum covering the whole unit circle
$\S^1$. See \cite[Sec.~2.1]{HJ_2014} for a similar result in the case $C_0=1_\H$
obtained by using the matrix representation of $S$ in the canonical basis of $\H$. See
also \cite{AP72,Tie17_1,RT19} for more general results on spectral properties of
unitary representations satisfying commutation relations similar to \eqref{eq_imp}.

Using what precedes, we can finally construct a conjugate operator for $U_0$ and
determine its spectral properties:

\begin{Proposition}[Spectral properties of $U_0$]\label{prop_spec_U_0}
Let
$$
A_0\Phi:=\left(\bigoplus_{k=1}^d\widetilde A\right)\Phi,
\quad\Phi\in\bigoplus_{k=1}^dC_{\rm c}(\Tau,\C^d).
$$
\begin{enumerate}[(a)]
\item $A_0$ is essentially self-adjoint in $\H_0$, with closure also denoted by $A_0$.
\item $U_0\in C^\infty(A_0)$ with $U_0^{-1}[A_0,U_0]=2\cdot1_{\H_0}$, and $U_0$
satisfies the imprimitivity relation
$$
\e^{isA_0}\gamma(U_0)\e^{-isA_0}=\gamma(\e^{2is}U_0),\quad s\in\R,~\gamma\in C(\S^1).
$$
\item $U_0$ is unitarily equivalent to a multiplication operator with purely
absolutely continuous spectrum
$$
\sigma(U_0)=\sigma_{\rm ac}(U_0)=\S^1.
$$
\end{enumerate}
\end{Proposition}

\begin{proof}
Point (a) follows from Lemma \ref{lemma_SC_0} and the fact that direct sums of
essentially self-adjoint operators are essentially self-adjoint. For point (b), using
Lemma \ref{lemma_SC_0} with $\widetilde U_0=SC_k$ ($k=1,\ldots,d$), we get the equalities
$$
U_0^{-1}[A_0,U_0]
=\bigoplus_{k=1}^d(SC_k)^{-1}[\widetilde A,SC_k]
=\bigoplus_{k=1}^d(2\cdot1_\H)
=2\cdot1_{\H_0}.
$$
This implies that $U_0\in C^\infty(A_0)$ with $U_0^{-1}[A_0,U_0]=2\cdot1_{\H_0}$, and
thus the imprimitivity relation
$$
\e^{isA_0}\gamma(U_0)\e^{-isA_0}=\gamma(\e^{2is}U_0),\quad s\in\R,~\gamma\in C(\S^1).
$$
Finally, using this relation and Mackey's imprimitivity theorem, one can show point
(c) as in the paragraph preceding this proposition.
\end{proof}

\subsection{Full evolution operator}\label{section_full}

In this section, we use the theory of Section \ref{sec_two} to construct a conjugate
operator $A$ for the full evolution operator $U$ starting from the conjugate operator
$A_0$ for $U_0$. As a by-product, we obtain a class of locally $U$-smooth operators
and determine spectral properties of $U$.

We start by showing that the perturbation $V=JU_0-UJ$ is trace class:

\begin{Lemma}\label{lemma_V_trace}
The perturbation $V$ factorises as $V=G^*G_0$, with $G_0\in S_2(\H_0)$ and
$G\in S_2(\H,\H_0)$. In particular, $V$ is trace class.
\end{Lemma}

\begin{proof}
A direct calculation gives for $\Phi=(\varphi_1,\dots,\varphi_d)\in\H_0$
$$
V\Phi
=\sum_{k=1}^d\big(\chi_kSC_k-SC\chi_k\big)\varphi_k
=\sum_{k=1}^d\big(S(C_k-C)\chi_k+[\chi_k,S]C_k\big)\varphi_k
=G^*G_0\Phi
$$
where
$$
G_0:=\bigoplus_{k=1}^d\langle\cdot\rangle^{-(1+\varepsilon_k)/2}
$$
and
$$
G^*:=DG_0\quad\hbox{with}\quad
D\Phi:=\sum_{k=1}^d\big(S(C_k-C)\chi_k+[\chi_k,S]C_k\big)
\langle\cdot\rangle^{1+\varepsilon_k}\varphi_k.
$$
Therefore, it is sufficient to show that $G_0\in S_2(\H_0)$ and $D\in\B(\H_0,\H)$ to
prove the claim.

To show that $G_0\in S_2(\H_0)$, it is sufficient to prove that
$\langle\cdot\rangle^{-s}\in S_2(\H)$ for $s>1/2$ since $\tfrac{1+\varepsilon_k}2>1/2$
for each $k=1,\ldots,d$. Let $(e_i)_{i=1}^d$ be the standard basis of $\C^d$. Then the
family $(\delta_x\otimes e_i)_{x\in\Tau,\,i=1,\ldots,d}$ is an orthonormal basis of $\H$,
and a direct calculation gives
$$
\big\|\langle\cdot\rangle^{-s}\big\|_{S_2(\H)}^2
=\sum_{x\in\Tau,\,i=1,\ldots,d}
\big\|\langle\cdot\rangle^{-s}(\delta_x\otimes e_i)\big\|_\H^2
=d\sum_{x\in\Tau}\langle x\rangle^{-2s}
<\infty.
$$

To show that $D\in\B(\H_0,\H)$, it is sufficient to prove the inclusions
$S(C_k-C)\chi_k\langle\cdot\rangle^{1+\varepsilon_k}\in\B(\H)$ and
$[\chi_k,S]C_k\langle\cdot\rangle^{1+\varepsilon_k}\in\B(\H)$ for $k=1,\ldots,d$. The first
inclusion follows directly from Assumption \ref{ass_short}. For the second inclusion,
we note that
$$
[\chi_k,S]C_k\langle\cdot\rangle^{1+\varepsilon_k}
=S
\left(\begin{smallmatrix}
S_{1+1,1+2}^{-1}[\chi_k,S_{1+1,1+2}] &&&\\
& S_{2+1,2+2}^{-1}[\chi_k,S_{2+1,2+2}] && \mbox{\large$0$}\\
\mbox{\large$0$} && \ddots &\\
&&& S_{d+1,d+2}^{-1}[\chi_k,S_{d+1,d+2}]\\
\end{smallmatrix}\right)
C_k\langle\cdot\rangle^{1+\varepsilon_k}
$$
with
\begin{align*}
S_{i,j}^{-1}[\chi_k,S_{i,j}]
&=\chi_{\rm e}\chi_k(\;\!\cdot\;\!a_j)+\chi_{\rm o}\chi_k(\;\!\cdot\;\!a_i)-\chi_k\\
&=\chi_{\rm e}\big(\chi_{\Tau_k\cdot a_j}+\delta_{k,1}\delta_{a_j}-\chi_{\Tau_k}\big)
+\chi_{\rm o}\big(\chi_{\Tau_k\cdot a_i}+\delta_{k,1}\delta_{a_i}-\chi_{\Tau_k}\big).
\end{align*}
But $\chi_{\Tau_k\cdot a_j}=\chi_{\Tau_k}+\delta_e-\delta_{a_k}$ if $j=k$ and
$\chi_{\Tau_k\cdot a_j}=\chi_{\Tau_k}$ if $j\ne k$. Therefore,
\begin{align*}
S_{i,j}^{-1}[\chi_k,S_{i,j}]
&=\chi_{\rm e}\big(\delta_{j,k}(\delta_e-\delta_{a_k})+\delta_{k,1}\delta_{a_j}\big)
+\chi_{\rm o}\big(\delta_{i,k}(\delta_e-\delta_{a_k})+\delta_{k,1}\delta_{a_i}\big)\\
&=\delta_{j,k}\delta_e-\delta_{i,k}\delta_{a_k}+\delta_{k,1}\delta_{a_i}\\
&=S_{i,j}^{-1}[\chi_k,S_{i,j}]\;\!\chi_{\{e,a_1,\dots,a_d\}}
\end{align*}
and thus
$$
[\chi_k,S]C_k\langle\cdot\rangle^{1+\varepsilon_k}
=[\chi_k,S]C_k\langle\cdot\rangle^{1+\varepsilon_k}\chi_{\{e,a_1,\dots,a_d\}}
\in\B(\H),
$$
which concludes the proof.
\end{proof}

Next, we show that the assumption (iv) of Theorem \ref{thm_rho_bis} is satisfied.

\begin{Lemma}\label{lemma_diff}
One has $\overline{VA_0\upharpoonright\dom(A_0)}\in\K(\H_0,\H)$.
\end{Lemma}

\begin{proof}
Take $\Phi=(\varphi_1,\dots,\varphi_d)\in\dom(A_0)$. Then it follows from the
proof of Lemma \ref{lemma_V_trace} that
$$
VA_0\Phi=\sum_{k=1}^d\big(S(C_k-C)\chi_k\widetilde A+F_k\big)\varphi_k,
$$
with $F_k:=[\chi_k,S]C_k\chi_{\{e,a_1,\dots,a_d\}}\widetilde A$ a finite rank operator.
In addition, we have
$$
S(C_k-C)\chi_k\widetilde A\;\!\varphi_k
=S(C_k-C)\chi_k\langle\cdot\rangle^{1+\varepsilon_k}
\langle\cdot\rangle^{-\varepsilon_k}
\widetilde A\;\!\langle\cdot\rangle^{-1}\varphi_k
$$
with $S(C_k-C)\chi_k\langle\cdot\rangle^{1+\varepsilon_k}$ bounded,
$\langle\cdot\rangle^{-\varepsilon_k}$ in the Schatten-class $S_\rho(\H)$ for any
$\rho>1/\varepsilon_k$, and $\widetilde A\langle\cdot\rangle^{-1}$ bounded. Indeed,
the operator $S(C_k-C)\chi_k\langle\cdot\rangle^{1+\varepsilon_k}$ is bounded due to
Assumption \ref{ass_short}, the operator $\langle\cdot\rangle^{-\varepsilon_k}$
belongs to $S_\rho(\H)$ because
$$
\big\|\langle\cdot\rangle^{-\varepsilon_k}\big\|_{S_\rho(\H)}^\rho
=\sum_{x\in\Tau,\,i=1,\ldots,d}\big\langle\langle\cdot\rangle^{-\rho\varepsilon_k}
(\delta_x\otimes e_i),(\delta_x\otimes e_i)\big\rangle_\H
=d\sum_{x\in\Tau}\langle x\rangle^{-\rho\varepsilon_k}
<\infty,
$$
and the operator $\widetilde A\langle\cdot\rangle^{-1}$ is bounded because
$$
\widetilde A\;\!\langle\cdot\rangle^{-1}
=\left(\begin{smallmatrix}
A_{1+1,1+2}\langle\cdot\rangle^{-1} &&&\\
& A_{2+1,2+2}\langle\cdot\rangle^{-1} && \mbox{\large$0$}\\
\mbox{\large$0$} && \ddots &\\
&&& A_{d+1,d+2}\langle\cdot\rangle^{-1}\\
\end{smallmatrix}\right)
$$
with
$$
A_{i,j}\langle\cdot\rangle^{-1}
=\big(\chi_{\rm e}\;\!|\cdot a_j|_{i,j}^2+\chi_{\rm o}\;\!|\cdot a_i|_{i,j}^2-|\cdot|_{i,j}^2\big)
\langle\cdot\rangle^{-1}\in\linf(\Tau)
$$
(this last inclusion can be verified as in point (ii) of the proof of Lemma
\ref{lemma_SC_0_ij}). It follows that
$$
VA_0\Phi=\sum_{k=1}^d(K_k+F_k)\varphi_k,
$$
with $K_k\in S_\rho(\H)$ and $F_k$ of finite rank. This implies that
$\overline{VA_0\upharpoonright\dom(A_0)}\in\K(\H_0,\H)$.
\end{proof}

We now define the conjugate operator $A$ for $U$ as in Section
\ref{sec_short}, and observe that in our case it coincides with the operator
$\widetilde A$ of Lemma \ref{lemma_SC_0}:

\begin{Lemma}[Conjugate operator for $U$]\label{lemma_A}
The operator
$$
A\;\!\varphi:=JA_0J^*\varphi,
\quad\varphi\in C_{\rm c}(\Tau,\C^d),
$$
is essentially self-adjoint in $\H$, with closure (also denoted by $A$) equal to
$\widetilde A$.
\end{Lemma}

\begin{proof}
Let $\varphi\in C_{\rm c}(\Tau,\C^d)$. Then the identity $\sum_{k=1}^d\chi_k^2\equiv1$
and the fact that diagonal multiplication operators mutually commute imply that
$$
A\;\!\varphi
=\tsum_{k=1}^d\chi_k\widetilde A\;\!\chi_k\varphi
=\widetilde A\;\!\big(\tsum_{k=1}^d\chi_k^2\big)\varphi
=\widetilde A\;\!\varphi.
$$
Since $C_{\rm c}(\Tau,\C^d)$ is a core for $\widetilde A$, it follows that
$A$ is essentially self-adjoint in $\H$, with closure equal to $\widetilde A$.
\end{proof}

By combining the results that precede, we can now establish a Mourre estimate for $U$.

\begin{Proposition}[Mourre estimate for $U$]\label{prop_Mourre}
One has $U\in C^1(A)$ and $\widetilde\varrho_U^A(\theta)\ge2$ for all $\theta\in\S^1$.
\end{Proposition}

\begin{proof}
Theorem \ref{thm_rho_bis} applies since its assumptions are verified in Proposition
\ref{prop_spec_U_0}(b) and Lemmas \ref{lemma_J}, \ref{lemma_V_trace},
\ref{lemma_diff}, and \ref{lemma_A}. Furthermore, Lemma \ref{lemma_properties}(c) and
Proposition \ref{prop_spec_U_0}(b) imply that
$\widetilde\varrho_{U_0}^{A_0}(\theta)\ge\varrho_{U_0}^{A_0}(\theta)=2$ for all
$\theta\in\S^1$. Thus, $U\in C^1(A)$ and
$\widetilde\varrho_U^A(\theta)\ge\widetilde\varrho_{U_0}^{A_0}(\theta)\ge2$ for all
$\theta\in\S^1$.
\end{proof}

To infer results for $U$ from the Mourre estimate of Proposition \ref{prop_Mourre},
one needs to verify a slightly stronger regularity condition than $U\in C^1(A):$

\begin{Lemma}\label{lemma_C_1_epsilon}
One has $U\in C^{1+\varepsilon}(A)$ for each $\varepsilon\in(0,1)$ with
$\varepsilon\le\min\{\varepsilon_1,\dots,\varepsilon_d\}$.
\end{Lemma}

\begin{proof}
We know from Proposition \ref{prop_Mourre} that $U\in C^1(A)$. To go further and prove
that $U\in C^{1+\varepsilon}(A)$, we need to show that
$$
\big\|\e^{-itA}[A,U]\e^{itA}-[A,U]\big\|_{\B(\H)}
\le{\rm Const.}\;\!t^\varepsilon\quad\hbox{for all $t\in(0,1)$.}
$$
Using the relations $S^{-1}[\widetilde A,S]=2\cdot1_\H$ and
$\sum_{k=1}^d\chi_k\equiv1$, we get for $\varphi\in C_{\rm c}(\Tau,\C^d)$
$$
[A,U]\varphi
=[\widetilde A,SC]\varphi
=\big([\widetilde A,S]C+S[\widetilde A,C]\big)\varphi
=\big(2U+\tsum_{k=1}^dS[\widetilde A,C\chi_k]\big)\varphi.
$$
For the second term, we get
$$
S[\widetilde A,C\chi_k]\varphi
=S[\widetilde A,(C-C_k)\chi_k]\varphi
=S[\widetilde A\;\!\langle\cdot\rangle^{-1},\langle\cdot\rangle(C-C_k)\chi_k]\varphi
$$
where $S\in C^1(A)$ with $[A,S]=S\big(S^{-1}[\widetilde A,S]\big)=2S$ and
\begin{equation}\label{def_D_k}
D_k:=[\widetilde A\;\!\langle\cdot\rangle^{-1},
\langle\cdot\rangle(C-C_k)\chi_k]\in\B(\H)
\end{equation}
due to Assumption \ref{ass_short} and the fact that
$\widetilde A\langle\cdot\rangle^{-1}\in\B(\H)$ (see the proof of Lemma
\ref{lemma_diff}). Since $C_{\rm c}(\Tau,\C^d)$ is a core for $A$, what precedes
implies that all the operators in the expression for $[A,U]$ belong to $C^1(A)$,
except the operators $D_k$ which we only know to be bounded. Therefore, we have to
show that
$$
\big\|\e^{-itA}D_k\e^{itA}-D_k\big\|_{\B(\H)}
\le{\rm Const.}\;\!t^\varepsilon\quad\hbox{for all $t\in(0,1)$.}
$$
Now, algebraic manipulations as presented in \cite[p.~325-326]{ABG96} show that for
all $t\in(0,1)$
\begin{align*}
\big\|\e^{-itA}D_k\e^{itA}-D_k\big\|_{\B(\H)}
&\le{\rm Const.}\;\!\big(\|\sin(tA)D_k\|_{\B(\H)}
+\|\sin(tA)(D_k)^*\|_{\B(\H)}\big)\\
&\le{\rm Const.}\;\!\big(\|tA\;\!(tA+i)^{-1}D_k\|_{\B(\H)}
+\|tA\;\!(tA+i)^{-1}(D_k)^*\|_{\B(\H)}\big).
\end{align*}
Furthermore, if we set $A_t:=tA\;\!(tA+i)^{-1}$ and
$\Lambda_t:=t\langle\cdot\rangle(t\langle\cdot\rangle+i)^{-1}$, we obtain that
$$
A_t=\big(A_t+i(tA +i)^{-1}A\;\!\langle\cdot\rangle^{-1}\big)\Lambda_t
$$
with $A\langle\cdot\rangle^{-1}\in\B(\H)$. Thus, since
$\|A_t+i(tA +i)^{-1}A\;\!\langle\cdot\rangle^{-1}\|_{\B(\H)}$ is bounded by a constant
independent of $t\in(0,1)$, it is sufficient to prove that
$$
\|\Lambda_t D_k\|_{\B(\H)}+\|\Lambda_t(D_k)^*\|_{\B(\H)}
\le{\rm Const.}\;\!t^\varepsilon\quad\hbox{for all $t\in(0,1)$.}
$$
But this estimate will hold if we show that the operators
$\langle\cdot\rangle^\varepsilon D_k$ and $\langle\cdot\rangle^\varepsilon(D_k)^*$
defined on $C_{\rm c}(\Tau,\C^d)$ extend continuously to elements of $\B(\H)$. For
this, we fix $\varepsilon\in(0,1)$ with
$\varepsilon\le\min\{\varepsilon_1,\dots,\varepsilon_d\}$, and note that
$\langle\cdot\rangle^{1+\varepsilon}(C-C_k)\chi_k\in\B(\H)$ due to Assumption
\ref{ass_short}. With this inclusion and the fact that
$\widetilde A\;\!\langle\cdot\rangle^{-1}\in\B(\H)$, one infers from \eqref{def_D_k}
that $\langle\cdot\rangle^\varepsilon D_k$ and
$\langle\cdot\rangle^\varepsilon(D_k)^*$ defined on $C_{\rm c}(\Tau,\C^d)$ extend
continuously to elements of $\B(\H)$, as desired.
\end{proof}

We are now in a position to obtain a class of locally $U$-smooth operators and
determine spectral properties of $U$.

\begin{Theorem}[Locally $U$-smooth operators]\label{thm_smooth_walk}
Let $\G$ be an auxiliary Hilbert space. Then each operator $T\in\B(\H,\G)$ which
extends continuously to an element of $\B\big(\dom(\langle\cdot\rangle^{-s}),\G\big)$
for some $s>1/2$ is locally $U$-smooth on any closed set
$\Theta'\subset\S^1\setminus\sigma_{\rm p}(U)$.
\end{Theorem}

\begin{proof}
We know from Proposition \ref{prop_Mourre} and Lemma \ref{lemma_C_1_epsilon} that for
each $\theta\in\S^1$ there exists an open set $\Theta_\theta\ni\theta$ for which the
assumptions of Theorem \ref{thm_locally_smooth} are satisfied. So, each operator
$T\in\B(\H,\G)$ which extends continuously to an element of
$\B\big(\dom(\langle A\rangle^s)^*,\G\big)$ for some $s>1/2$ is locally $U$-smooth on
any closed set $\Theta_\theta'\subset\Theta_\theta\setminus\sigma_{\rm p}(U)$. Since
any closed set $\Theta'\subset\S^1\setminus\sigma_{\rm p}(U)$ is contained in a finite
union of closed sets of type $\Theta_\theta'$, we infer that $T$ is locally $U$-smooth
on $\Theta'$ too.

Now, we know from the proof of of Lemma \ref{lemma_diff} that
$\dom(\langle\cdot\rangle)\subset\dom(A)$. Therefore, we have
$\dom(\langle\cdot\rangle^s)\subset\dom(\langle A\rangle^s)$ for each $s>1/2$, and it
follows by duality that
$
\dom(\langle A\rangle^s)^*
\subset\dom(\langle\cdot\rangle^s)^*
\equiv\dom(\langle\cdot\rangle^{-s})
$
for each $s>1/2$. In consequence, any operator $T\in\B(\H,\G)$ which extends
continuously to an element of $\B\big(\dom(\langle\cdot\rangle^{-s}),\G\big)$ for some
$s>1/2$ also extends continuously to an element of
$\B\big(\dom(\langle A\rangle^s)^*,\G\big)$. This concludes the proof.
\end{proof}

\begin{Theorem}[Spectral properties of $U$]\label{thm_spec_U}
The operator $U$ has at most finitely many eigenvalues, each one of finite
multiplicity, and no singular continuous spectrum.
\end{Theorem}

\begin{proof}
Let $\theta\in\S^1$. Then we know from Proposition \ref{prop_Mourre} and Lemma
\ref{lemma_C_1_epsilon} that there exists an open set $\Theta_\theta\ni\theta$ for
which the assumptions of Theorem \ref{thm_spec_prop} are satisfied. Thus $U$ has at
most finitely many eigenvalues in $\Theta_\theta$, each one of finite multiplicity,
and $U$ has no singular continuous spectrum in $\Theta_\theta$. Since $\S^1$ can be
covered by a finite number of open sets of type $\Theta_\theta$, it follows that $U$
has at most finitely many eigenvalues, each one of finite multiplicity, and no
singular continuous spectrum.
\end{proof}

\subsection{Wave operators}\label{sec_wave}

In this final section, we use the results obtained so far to establish the existence
and completeness of the wave operators for the triple $(U,U_0,J)$. We also explain why
at least two operators different from $U_0$ can be used as a free evolution operator
for $U$, and we establish the existence and completeness in these cases too.

\begin{Theorem}[Completeness, version 1]\label{thm_comp_1}
The wave operators $W_\pm(U,U_0,J):\H_0\to\H$ given by
$$
W_\pm(U,U_0,J):=\slim_{n\to\pm\infty}U^{-n}JU_0^n
$$
exist and are complete, that is, $\Ran\big(W_\pm(U,U_0,J)\big)=\H_{\rm ac}(U)$.
\end{Theorem}

\begin{proof}
We know from Lemma \ref{lemma_V_trace} that $V$ is trace class. Thus, it follows from
\cite[Ex.~3.8]{Tie_2020} that $W_\pm(U,U_0,J)$ exist. Similarly, since
$J^*U-U_0J^*=U_0V^*U$ is trace class too, the wave operators
$$
W_\pm(U_0,U,J^*):=\slim_{n\to\pm\infty}U_0^{-n}J^*U^nP_{\rm ac}(U)
$$
exist too. Now, let $a_\pm,b_\pm\in\B(\H)$ be defined by
$$
a_\pm\varphi:=W_\pm(U,U_0,J)(\varphi,\dots,\varphi)
\quad\hbox{and}\quad b_\pm\varphi:=\sum_{k=1}^d\big(W_\pm(U_0,U,J^*)\varphi\big)_k,
\quad\varphi\in\H.
$$
Since $\H_{\rm ac}(U)\supset\Ran\big(W_\pm(U,U_0,J)\big)\supset\Ran(a_\pm)$, it is
sufficient to show that $\Ran(a_\pm)=\H_{\rm ac}(U)$ to conclude the proof. To achieve
this, we need to recall some information: $U_0$ has purely absolutely continuous
spectrum by Lemma \ref{prop_spec_U_0}(c) and $U$ has at most finitely many eigenvalues
and no singular continuous spectrum by Theorem \ref{thm_spec_U}. So
$1_{\H_0}=E^{U_0}(\S^1\setminus\sigma_{\rm p}(U))$ and
$P_{\rm ac}(U)=E^U(\S^1\setminus\sigma_{\rm p}(U))$, and we have for all
$\varphi,\psi\in\H$ the equalities
\begin{align*}
\big\langle a_\pm^*\varphi,\psi\big\rangle_\H
&=\big\langle\varphi,P_{\rm ac}(U)W_\pm(U,U_0,J)(\psi,\dots,\psi)\big\rangle_\H\\
&=\lim_{n\to\pm\infty}\sum_{k=1}^d\big\langle\varphi,
P_{\rm ac}(U)U^{-n}\chi_kU_k^n\psi\big\rangle_\H\\
&=\lim_{n\to\pm\infty}\sum_{k=1}^d\big\langle U_k^{-n}\chi_kU^nP_{\rm ac}(U)\varphi,
\psi\big\rangle_\H\\
&=\left\langle\sum_{k=1}^d\big(W_\pm(U_0,U,J^*)\varphi\big)_k,\psi\right\rangle_\H\\
&=\big\langle b_\pm\varphi,\psi\big\rangle_\H.
\end{align*}
Thus $b_\pm$ is the adjoint of $a_\pm$. Furthermore, since
$$
a_\pm
=\slim_{n\to\pm\infty}U^{-n}\tsum_{k=1}^d\chi_kU_k^n
=\slim_{n\to\pm\infty}U^{-n}\big(\tsum_{k=1}^d\chi_kC_k\big)^nS^n
$$
with $\tsum_{k=1}^d\chi_kC_k$ unitary, the operator $a_\pm$ is an isometry and thus
has closed range. Also, since
$$
b_\pm
=\slim_{n\to\pm\infty}\tsum_{k=1}^dU_k^{-n}\chi_kU^nP_{\rm ac}(U)
=\slim_{n\to\pm\infty}S^{-n}\big(\tsum_{k=1}^d\chi_kC_k\big)^{-n}U^n
E^U(\S^1\setminus\sigma_{\rm p}(U)),
$$
we have $\ker(b_\pm)=E^U(\sigma_{\rm p}(U))$. Combining what precedes, we obtain that
$$
\Ran(a_\pm)
=\H\ominus\ker(a_\pm^*)
=E^U(\S^1)\H\ominus\ker(b_\pm)
=E^U(\S^1\setminus\sigma_{\rm p}(U))\H
=\H_{\rm ac}(U),
$$
as desired.
\end{proof}

The proof of Theorem \ref{thm_comp_1} implies in particular that:

\begin{Corollary}[Completeness, version 2]\label{cor_comp_2}
The operators $a_\pm:\H\to\H$ given by
$$
a_\pm:=\slim_{n\to\pm\infty}U^{-n}J_nS^n,
\quad J_n:=\big(\tsum_{k=1}^d\chi_kC_k\big)^n,
$$
exist, are isometric, and satisfy $\Ran(a_\pm)=\H_{\rm ac}(U)$.
\end{Corollary}

Corollary \ref{cor_comp_2} means that the operators $a_\pm$ are complete wave
operators, with time-dependent identification operators $J_n$, for the pair $(U,S)$.
This shows that if one uses the matrix powers $\big(\tsum_{k=1}^d\chi_kC_k\big)^n$ as
time-dependent identification operators, then the (pretty simple) shift $S$ can be
chosen as a free evolution operator for $U$.

\begin{Corollary}\label{cor_spec_U}
One has $\sigma_{\rm ac}(U)=\S^1$.
\end{Corollary}

\begin{proof}
Let $U_+:=(a_+)^*Ua_+$. Since $a_+$ is unitary from $\H$ to $\H_{\rm ac}(U)$, we have
$$
\sigma_{\rm ac}(U)
=\sigma\big(U\upharpoonright\H_{\rm ac}(U)\big)
=\sigma\big((a_+)^*Ua_+\big)
=\sigma(U_+).
$$
So it is sufficient to show that $\sigma(U_+)=\S^1$ to prove the claim. Let
$\varphi\in\H$. Then the definition of $a_+$, the intertwining property of
$W_+(U,U_0,J)$, and Lemma \ref{lemma_J} imply that
\begin{align*}
U_+\varphi
&=\sum_{k=1}^d\big(W_+(U_0,U,J^*)W_+(U,U_0,J)U_0(\varphi,\dots,\varphi)\big)_k\\
&=\slim_{n\to\infty}\sum_{k=1}^d\big(U_0^{-n}J^*JU_0^n
(U_1\varphi,\dots,U_d\varphi)\big)_k\\
&=\slim_{n\to\infty}S^{-n}\big(\tsum_{k=1}^d\chi_kC_k)S^{n+1}\varphi.
\end{align*}
This, together with the fact that $[\widetilde A,\tsum_{k=1}^d\chi_kC_k]=0$ and the
imprimitity relation \eqref{eq_imp}, implies for $s\in\R$
\begin{align*}
\e^{is\widetilde A}U_+\e^{-is\widetilde A}
&=\slim_{n\to\infty}\big(\e^{is\widetilde A}S\e^{-is\widetilde A}\big)^{-n}
\big(\tsum_{k=1}^d\chi_kC_k)\big(\e^{is\widetilde A}S\e^{-is\widetilde A}\big)^{n+1}\\
&=\slim_{n\to\infty}(\e^{2is}S)^{-n}\big(\tsum_{k=1}^d\chi_kC_k)(\e^{2is}S)^{n+1}\\
&=\e^{2is}U_+.
\end{align*}
So $U_+$ is unitarily equivalent to $\e^{2is}U_+$ for each $s\in\R$, and thus
$\sigma(U_+)=\S^1$.
\end{proof}

Combining Theorem \ref{thm_spec_U} and Corollary \ref{cor_spec_U}, we infer that the
spectrum of $U$ covers the whole unit circle and is purely absolutely continuous,
outside possibly a finite set where $U$ may have eigenvalues of finite multiplicity.

Interestingly enough, a third choice of free evolution operator is yet possible for
the proof of the existence and completeness of the wave operators, namely, the
operator $\widetilde U_0:=S\sum_{k=1}^d\chi_kC_k:$

\begin{Theorem}[Completeness, version 3]\label{thm_comp_3}
The wave operators $W_\pm(U,\widetilde U_0):\H\to\H$ given by
$$
W_\pm(U,\widetilde U_0):=\slim_{n\to\pm\infty}U^{-n}(\widetilde U_0)^n
$$
exist, are isometric, and are complete, that is,
$\Ran\big(W_\pm(U,\widetilde U_0)\big)=\H_{\rm ac}(U)$.
\end{Theorem}

\begin{proof}
The coin operator $C_0:=\sum_{k=1}^d\chi_kC_k$ is diagonal. So we know from Section
\ref{section_free} that $\widetilde U_0$ has purely absolutely continuous spectrum
covering the whole unit circle $\S^1$. Furthermore, calculations similar to that of
Lemma \ref{lemma_V_trace} show that the perturbation
$
\widetilde U_0-U=S\sum_{k=1}^d\chi_k(C_k-C)
$
is trace class. Therefore, we obtain from \cite[Ex.~3.8]{Tie_2020} that the wave
operators
$$
W_\pm(U,\widetilde U_0)
\quad\hbox{and}\quad
W_\pm(\widetilde U_0,U):=\slim_{n\to\pm\infty}(\widetilde U_0)^{-n}U^nP_{\rm ac}(U)
$$
exist. And the operators $W_\pm(U,\widetilde U_0)$ are isometric. Since
$1_\H=E^{\widetilde U_0}(\S^1\setminus\sigma_{\rm p}(U))$ and
$P_{\rm ac}(U)=E^U(\S^1\setminus\sigma_{\rm p}(U))$, it follows from Remark
\ref{rem_unitary} and Lemma \ref{lemma_J_complete} that
$$
\Ran\big(W_\pm(U,\widetilde U_0)\big)
=E^U(\S^1\setminus\sigma_{\rm p}(U))\H
=\H_{\rm ac}(U),
$$
as desired.
\end{proof}

\begin{Remark}\label{rem_3_free}
Each choice of free evolution operator comes with its pros and cons. The shift $S$ is
a very simple operator, but the corresponding complete, isometric, wave operators
$a_\pm$ are not easily interpretable since they are defined with time-dependent
identification operators. The operator $\widetilde U_0=S\sum_{k=1}^d\chi_kC_k$ is
simple and the corresponding complete, isometric, wave operators
$W_\pm(U,\widetilde U_0)$ are intuitive. But $\widetilde U_0$ does not encode in
separate Hilbert spaces the multichannel structure of the scattering system. The $d$
scattering channels are summed together in the Hilbert space $\H$ via the matrix
$\sum_{k=1}^d\chi_kC_k$. Finally, the operator $U_0$ is simple, encodes in separate
Hilbert spaces the multichannel structure of the scattering system, and the
corresponding complete wave operators $W_\pm(U,U_0,J)$ are intuitive. But the
operators $W_\pm(U,U_0,J)$ are not isometric. We expect them to be partial isometries
with nontrivial initial subspaces $\H_0^\pm\subset\H_0$ defined in terms of some
asymptotic velocity operators, as in the case of anisotropic quantum walks on $\Z$
(see the discussion below).
\end{Remark}

To conclude, we list some interesting generalisations and problems left open by this
work on quantum walks on trees:

\begin{enumerate}[(i)]
\item As mentioned in Remark \ref{rem_3_free}, it would be interesting to determine
the initial subspaces $\H_0^\pm\subset\H_0$ on which the wave operators
$W_\pm(U,U_0,J)$ are isometric. As for anisotropic quantum walks on $\Z$
\cite[Prop.~3.4]{RST_2019}, we expect these subspaces to be defined in terms of some
asymptotic velocity operators for the operators $U_i=SC_i$ ($i=1,\ldots,d$). However,
on trees of odd degree $d\ge3$, there are no canonical position operators coming to
mind allowing to define straightforwardly asymptotic velocity operators. So, one might
have to do something new in order to determine the subspaces $\H_0^\pm$.
\item It would be interesting to consider the case of quantum walks with coin operator
converging at infinity to constant coins along partitions of the tree $\Tau$ different
from the one considered here. In this first work on the topic, we used for simplicity
the partition of $\Tau$ into its $d$ main branches, but many other more refined
choices of partitions into subtrees are possible.
\item In the case $d=3$, for evolution operators $U=SC$ with constant coin operator
$C\in\U(3)$, it has been shown in \cite[Sec.~2]{JM_2014} that one can impose boundary
conditions that restrict the configuration space of the quantum walker to a rooted
tree, while preserving the unitarity of $U$. It would be interesting to generalise
this restriction procedure to evolution operators with non-constant coin operators and
determine which results of this paper still hold in that modified setup.
\item In this work, we have treated coin operators that converge on each main branch
of $\Tau$ to a constant diagonal unitary matrix, whereas in the case of quantum walks
on $\Z$ the authors of \cite{RST_2018,RST_2019} have covered coin operators that
converge at infinity to arbitrary constant unitary matrices. This is explained by the
fact that on $\Tau$ there isn't any explicit Fourier transform allowing to diagonalise
non-diagonal asymptotic operators. So, the construction of a conjugate operator based
on the Fourier transform presented in \cite[Sec.~4.1]{RST_2018} is not possible here.
That being said, each cyclic subspace of $\ell^2(\Tau)$ of the form
$\overline{\Span}\{S_{i,j}^n\delta_x\mid n\in\Z\}$ is isomorphic to $\ell^2(\Z)$, and
thus admits a Fourier transform \cite[Lemma~2.2]{HJ_2014}. Therefore, it would be
worth investigating if these implicit Fourier transforms can be used to construct a
conjugate operator suitable for quantum walks on $\Tau$ with coin operators converging
on each main branch of $\Tau$ to an arbitrary constant unitary matrix.
\end{enumerate}



\end{document}